\documentclass[12pt]{amsart}
\usepackage{}

\usepackage{amsmath}
\usepackage{amsfonts}
\usepackage{amssymb}

\usepackage[all,cmtip]{xy}           
\usepackage{bm}
\usepackage{bbm}

\usepackage{bbding}
\usepackage{txfonts}
\usepackage{amscd}
\usepackage[all]{xy}
\usepackage[shortlabels]{enumitem}
\usepackage{ifpdf}
\ifpdf
  \usepackage[colorlinks,final,backref=page,hyperindex]{hyperref}
\else
  \usepackage[colorlinks,final,backref=page,hyperindex,hypertex]{hyperref}
\fi
\usepackage{tikz}
\usepackage[active]{srcltx}
\usepackage{caption}

\makeatletter

\newtheorem{thm}{Theorem}[section]
\newtheorem{lem}[thm]{Lemma}
\newtheorem{cor}[thm]{Corollary}
\newtheorem{pro}[thm]{Proposition}
\newtheorem{ex}[thm]{Example}
\newtheorem{rmk}[thm]{Remark}
\newtheorem{defi}[thm]{Definition}

\newcommand {\emptycomment}[1]{}

\newcommand{\be }{\begin{equation}}
\newcommand{\ee }{\end{equation}}

\newcommand{\g}{\mathfrak g}
\newcommand{\h}{\mathfrak h}

\newcommand{\huaB}{\mathcal{B}}

\newcommand{\huaL}{\mathcal{L}}
\newcommand{\huaR}{\mathcal{R}}

\newcommand{\huaP}{\mathcal{P}}

\newcommand{\huaI}{\mathcal{I}}

\newcommand{\huaO}{{\mathcal{O}}}

\newcommand{\frkd}{\mathfrak d}

\newcommand{\Id}{{\rm{Id}}}

\newcommand{\br}[1]{   [ \cdot,    \cdot  ]   }

\newcommand{\End}{\mathrm{End}}

\newcommand{\gl}{\mathfrak {gl}}

\newcommand{\Img}{\mathrm{Im}}

\begin{document}

\title{Phase space of a Poisson algebra and the induced Pre-Poisson bialgebra}

\author{You Wang}
\address{Department of Mathematics, Jilin University, Changchun 130012, Jilin, China}
\email{wangyou20@mails.jlu.edu.cn}

\author{Yunhe Sheng}
\address{Department of Mathematics, Jilin University, Changchun 130012, Jilin, China}
\email{shengyh@jlu.edu.cn}


\begin{abstract}
    In this paper, we first introduce the notion of a phase space of a Poisson algebra, and show that a Poisson algebra has a phase space if and only if it is sub-adjacent to a pre-Poisson algebra. Moreover, we introduce the notion of Manin triples of pre-Poisson algebras and show that there is a one-to-one correspondence between Manin triples of pre-Poisson algebras and phase spaces of Poisson algebras. Then we introduce the notion of pre-Poisson bialgebras, which is equivalent to Manin triples of pre-Poisson algebras. We study   coboundary pre-Poisson bialgebras, which leads to an analogue of the classical Yang-Baxter equation. Furthermore, we introduce the notions of quasi-triangular and factorizable pre-Poisson bialgebras as special cases. A quasi-triangular pre-Poisson bialgebra gives rise to a relative Rota-Baxter operator of weight $1$. The double of a pre-Poisson bialgebra enjoys a natural factorizable pre-Poisson bialgebra structure. Finally, we introduce the notion of quadratic Rota-Baxter pre-Poisson algebras and show that there is a one-to-one correspondence between quadratic Rota-Baxter pre-Poisson algebras and factorizable pre-Poisson bialgebras. Based on this construction, we give a phase space for a Rota-Baxter symplectic Poisson algebra.
\end{abstract}

\renewcommand{\thefootnote}{}


\maketitle

\tableofcontents

\allowdisplaybreaks


\section{Introduction}
The purpose of this paper is to study phase spaces of Poisson algebras in terms of pre-Poisson algebras and pre-Poisson bialgebras. Moreover, we also develop the quasi-triangular and factorizable theories for pre-Poisson bialgebras.

\subsection{Poisson algebras and pre-Poisson algebras}

Poisson algebras were originally coming from the study of the Hamiltonian mechanics and then appeared in many fields in mathematics and mathematical physics. In particular, Poisson algebras are the algebraic structures corresponding to Poisson manifolds.
The notion of a pre-Poisson algebra was first introduced by Aguiar (\cite{Aguiar}), which combines a Zinbiel algebra and a pre-Lie algebra on the same vector space satisfying some compatibility conditions. The notion of Zinbiel algebras (also called dual Leibniz algebras) was introduced by Loday in \cite{Lod1}, and further studied in \cite{Liv,Lod2}. Pre-Lie
algebras are a class of nonassociative algebras coming from the study of convex homogeneous cones, affine manifolds and affine
structures on Lie groups, and  cohomologies of associative algebras.  They also appeared in many fields in
mathematics and mathematical physics, such as complex and symplectic structures on Lie groups and Lie algebras, integrable
systems, Poisson brackets and infinite dimensional Lie algebras, vertex algebras, quantum field theory and operads. See the survey
\cite{Bai-Review,Bu} for more details.  A pre-Poisson algebra naturally gives rise to its sub-adjacent Poisson algebra through the anti-commutator of Zinbiel algebras and the commutator of pre-Lie algebras. On the other hand, A Rota-Baxter operator of weight zero on a Poisson algebra can give rise to a pre-Poisson algebra (\cite{Aguiar}).
In \cite{LMS}, the authors introduced the notion of Hom-pre-Poisson algebras and discuss the relation with Hom-Poisson algebras.

\subsection{Phase spaces of Lie algebras, Manin triples of pre-Lie algebras and pre-Lie bialgebras}

In classical mechanics, the phase space is a symplectic manifold of positions and momenta. The phase space over a vector space represents an extension of this concept. The phase space of a vector space $V$ is the direct sum $T^*(V)=V\oplus V^*$ together with the natural nondegenerate skew-symmetric bilinear form
\begin{equation}\label{natural-skew-sym-form}
\omega(x+\xi,y+\eta)=\langle \xi,y \rangle-\langle\eta,x\rangle, \quad \forall x,y\in V,\xi,\eta\in V^*,
\end{equation}
where $\langle \cdot,\cdot \rangle$ is the ordinary pairing between $V$ and $V^*$. The notion of phase spaces of Lie algebras was introduced by   Kupershmidt (\cite{Ku-phase-space}). A non-abelian phase space of a Lie algebra $\g$ is a Lie algebra $T^*(\g)=\g \oplus \g^*$ together with the symplectic form $\omega$ given by \eqref{natural-skew-sym-form} such that $\g$ and $\g^*$ are Lie subalgebras. In \cite{Ku-phase-space,Ku2},   Kupershmidt studied phase spaces of certain Lie algebras, such as $\gl(V)$, vector fields on $\mathbf{R}^n$, the current algebras and the Virasoro algebras.

In \cite{Bai-phase-space}, Bai proved that a Lie algebra has a phase space if and only if it is sub-adjacent to a pre-Lie algebra. In \cite{Bai}, Bai introduced the notions of Manin triples of pre-Lie algebras and pre-Lie bialgebras, and showed that both of them are equivalent to the phase spaces of Lie algebras. A bialgebra theory for some algebraic structure usually contains two same algebraic structures on a vector space and its dual space satisfying certain compatibility conditions. Usually,  Manin triples and bialgebras are equivalent. There are many important bialgebras, such as Lie bialgebras (\cite{Drinfeld,Kosmann}), Hopf algebras (\cite{Mon}), infinitesimal bialgebras (\cite{Aguiar3,Bai-asso,Zhe}) and Poisson bialgebras (\cite{NB}). Certain subclasses of bialgebras, like quasi-triangular ones, triangular ones and factorizable ones, play important roles in the whole theory. In the context of Lie bialgebras, classical $r$-matrices give rise to  quasi-triangular Lie bialgebras (\cite{STS}). Moreover, factorizable Lie bialgebras, which are special quasi-triangular Lie bialgebras,  have important applications in integrable systems (\cite{RS,S2}). In \cite{HB}, the authors found that solutions of the S-equation in a pre-Lie algebra, which is an analogue of the classical Yang-Baxter equation can induce phase spaces of Lie algebras and all these phase spaces have a symplectically isomorphic property.



\subsection{Rota-Baxter characterization of factorizable bialgebras}

The notion of Rota-Baxter algebras was initiated by G. Baxter (\cite{Ba}) in his probability study to understand Spitzers identity in fluctuation theory. Rota-Baxter operators are broadly connected with mathematical physics, including the application of Connes-Kreimer's algebraic approach in the renormalization of perturbative quantum field theory (\cite{CK}). Rota-Baxter algebras are also closely related to double algebras (\cite{G3,Goncharov}), see \cite{Guo} for more details.  In the Lie algebra context, a Rota-Baxter operator was introduced independently in the 1980s as the operator form of the classical Yang-Baxter equation  (\cite{STS}). In order to better understand the relationship between the classical Yang-Baxter equation and the related integrable systems, Kupershmidt introduced an $\huaO$-operator on a Lie algebra (\cite{Ku}) as a more general notion.

Note that Rota-Baxter operators can be used to characterize certain bialgebras.  Goncharov established a correspondence between non-skewsymmetric solutions of the classical Yang-Baxter equation and  Rota-Baxter operators of nonzero weight for certain Lie algebras (\cite{G1,G2}). See \cite{BGN} for more general study. In \cite{Lang}, it showed that there is a one-to-one correspondence between factorizable Lie bialgebras and quadratic Rota-Baxter Lie algebras of nonzero weight. In \cite{WBLS}, factorizable pre-Lie bialgebras are characterized by quadratic Rota-Baxter pre-Lie algebras. In \cite{Wang}, factorizable Zinbiel bialgebras are related to  quadratic Rota-Baxter Zinbiel algebras.

\subsection{Main results and outline of the paper}

In this paper, we study phase spaces of a Poisson algebra $P$, which is a symplectic Poisson algebra structure on $P\oplus P^*$ such that $P$ and $P^*$ are Poisson subalgebras and the symplectic structure is given by \eqref{natural-skew-sym-form}. We show that a Poisson algebra has a phase space if and only if it is sub-adjacent to a pre-Poisson algebra. We introduce the notions of Manin triples of pre-Poisson algebras and enhance the above characterization to the equivalence between phase spaces of Poisson algebras and Manin triples of pre-Poisson algebras. Guided by the general philosophy for the equivalence between Manin triples and bialgebras, we further develop the theory of pre-Poisson bialgebras, and show that Manin triples of Poisson algebras and pre-Poisson bialgebras are equivalent. Quasi-triangular and factorizable pre-Poisson bialgebras are introduced and the former give rise to   relative Rota-Baxter operators of weight $1$.
We also give the Rota-Baxter characterization of factorizable pre-Poisson bialgebras. We introduce the notion of quadratic Rota-Baxter pre-Poisson algebras of weight $\lambda$ and show that there is a one-to-one correspondence between factorizable pre-Poisson  bialgebras and quadratic Rota-Baxter pre-Poisson algebras of nonzero weight. Consequently, a Rota-Baxter symplectic Poisson algebra naturally admits a phase space which is the subadjacent Poisson algebra of the double of a pre-Poisson bialgebra.

The paper is organized as follows. In Section \ref{sec:phase-space}, we introduce the notion of a phase space of a Poisson algebra, and characterize it using pre-Poisson algebras and Manin triples of pre-Poisson algebras. In Section \ref{sec:pre-Poisson-bialgebras}, we introduce the notion of a pre-Poisson bialgebra, which is  equivalent to a Manin triple of pre-Poisson algebras. In Section \ref{sec:quasi-fac}, we introduce the notions of quasi-triangular and factorizable pre-Poisson bialgebras as special coboundary pre-Poisson bialgebras. 
In Section \ref{sec:RB-fac}, we introduce the notion of quadratic Rota-Baxter pre-Poisson algebras and show that there is a one-to-one correspondence between quadratic Rota-Baxter pre-Poisson algebras and factorizable pre-Poisson bialgebras.

Through out this paper, we work over a field $\mathbb{K}$ and all the vector spaces and algebras are over $\mathbb{K}$ and finite-dimensional.

\section{The phase spaces of Poisson algebras and Manin triples of pre-Poisson algebras}\label{sec:phase-space}

In this section, we introduce the notions of a symplectic structure and a phase space of a Poisson algebra. Then we show that a Poisson algebra has a phase space if and only if it is sub-adjacent to a pre-Poisson algebra. Moreover, we introduce the notion of Manin triples of pre-Poisson algebras and show that there is one-to-one correspondence between phase spaces of Poisson algebras and Manin triples of pre-Poisson algebras.

First we recall the notion of Poisson algebras and pre-Poisson algebras.

\begin{defi}
A {\bf Poisson algebra} is a triple $(P,\cdot_P,\{\cdot,\cdot\}_P)$, where $(P,\cdot_P)$ is a commutative associative algebra and $(P,\{\cdot,\cdot\}_P)$ is a Lie algebra, such that the Leibniz rule holds:
$$ \{x,y\cdot_{P} z \}_P=\{x,y \}_P \cdot_{P} z+ y\cdot_{P} \{x, z \}_P, \quad \forall x,y,z\in P.$$
\end{defi}

Let $(P,\cdot_P,\{\cdot,\cdot\}_P)$ be a Poisson algebra, $V$ a vector space and $\varrho,\varsigma$ two linear maps. Recall that $(V;\varrho,\varsigma)$ is called a {\bf representation of a Poisson algebra} $(P,\cdot_P,\{\cdot,\cdot\}_P)$, if $(V;\varrho)$ is a representation of the Lie algebra $(P,\{\cdot,\cdot\}_P)$ and $(V;\varsigma)$ is a representation of the commutative associative algebra $(P,\cdot_P)$ and two linear maps $\varrho,\varsigma$ satisfy the following conditions:
\begin{eqnarray*}
\varrho(x \cdot_{P} y)&=&\varsigma(y)\varrho(x)+\varsigma(x)\varrho(y);\\
\varsigma(\{x,y\}_P)&=& \varrho(x)\varsigma(y)-\varsigma(y)\varrho(x), \quad \forall x,y\in P.
\end{eqnarray*}

In fact, $(V;\varrho,\varsigma)$ is a representation of a Poisson algebra $(P,\cdot_P,\{\cdot,\cdot\}_P)$ if and only if the direct sum $P\oplus V$ of vector spaces is a Poisson algebra with the multiplication $\cdot_{\ltimes}$ and $\{\cdot,\cdot\}_{\ltimes}$ given by
\begin{eqnarray*}
(x+u)\cdot_{\ltimes}(y+v)&=&x \cdot_P y+\varsigma(x)v+\varsigma(y)u;\\
\{x+u,y+v\}_{\ltimes}&=&\{x,y\}_P+\varrho(x)v-\varrho(y)u,
\end{eqnarray*}
for all $x,y\in P,u,v\in V$. We denote it by $P\ltimes_{\varrho,\varsigma} V$.

\begin{defi}{\rm(\cite{Aguiar,Loday})}
A {\bf Zinbiel algebra} $(A,\ast_A)$ is a vector space $A$ together with a bilinear operation $\ast_A:A \otimes A \to A$ such that
\begin{equation}\label{Zin-identity}
x\ast_A (y \ast_A z)=(x \ast_A y+y\ast_A x)\ast_A z, \quad \forall x,y,z \in A,
\end{equation}
or equivalently,
$$  (x,y,z)=(y\ast_A x)\ast_A z, \quad \forall x,y,z \in A, $$
where $(x,y,z)=x\ast_A (y \ast_A z)-(x \ast_A y)\ast_A z$ is the associator.
\end{defi}

Let $(A,\ast_A)$ be a Zinbiel algebra. Then the bilinear operation $\cdot_A:A \otimes A \to A$ defined by
\begin{equation}\label{Zin-com}
x \cdot_A y=x \ast_A y+y\ast_A x, \quad \forall  x,y \in A
\end{equation}
defines a commutative associative algebra $(A,\cdot_A)$, which is called the {\bf sub-adjacent commutative associative algebra} of $(A,\ast_A)$ and denoted by $A^c$. In this situation, $L_\ast:A \to \End(A)$ defined by $L_\ast (x)(y)=x\ast_A y$ gives a representation of the commutative associative algebra $A^c$ on $A$.

\begin{defi}
A {\bf pre-Lie algebra } $(\g,\circ_{\g})$ is a vector space $\g$ equipped with a bilinear multiplication $\circ_{\g}:\g\otimes \g \longrightarrow \g $ such that for any $x,y,z\in \g$, the associator $(x,y,z)=(x\circ_{\g}y)\circ_{\g}z-x\circ_{\g}(y\circ_{\g}z)$ is symmetric in $x,y,~i.e.$  $(x,y,z)=(y,x,z)$, or equivalently,
\begin{equation}\label{preLie-identity}
 (x\circ_{\g}y)\circ_{\g}z-x\circ_{\g}(y\circ_{\g}z)=(y\circ_{\g}x)\circ_{\g}z-y\circ_{\g}(x\circ_{\g}z).
\end{equation}
\end{defi}

Let $(\g,\circ_{\g})$ be a pre-Lie algebra. Then the commutator
\begin{equation}\label{preLie-com}
  [x,y]_{\g}:=x\circ_{\g}y-y\circ_{\g}x,\quad\forall x,y\in \g,
\end{equation}
defines a Lie algebra $(\g,[\cdot,\cdot]_{\g})$, which is called the {\bf sub-adjacent Lie algebra} of $(\g,\circ_{\g})$ and denoted by $\g^{c}$. Furthermore, $L_\circ:\g\rightarrow
\gl(\g)$ defined by $L_\circ (x)(y)=x\circ_\g y$ gives a representation of the Lie
algebra $\g^c$ on $\g$.

\begin{defi}{\rm(\cite{Aguiar})}
A {\bf pre-Poisson algebra} is a triple $(P,\ast_P,\circ_P)$, where $(P,\ast_P)$ is a Zinbiel algebra and
$(P,\circ_P)$ is a pre-Lie algebra, satisfying the following compatibility conditions:
\begin{eqnarray}
\label{pre-poisson1}(x \circ_P y-y \circ_P x)\ast_P z &=& x \circ_P (y \ast_P z)-y \ast_P (x \circ_P z);\\
\label{pre-poisson2}(x \ast_P y+y \ast_P x)\circ_P z &=& x\ast_P (y \circ_P z)+y \ast_P (x \circ_P z).
\end{eqnarray}
\end{defi}

\begin{pro}{\rm(\cite{Aguiar})}
Let $(P,\ast_P,\circ_P)$ be a pre-Poisson algebra. Then $(P,\cdot_P,\{\cdot,\cdot\}_P)$ is a Poisson algebra, where the Lie bracket $\{\cdot,\cdot\}_P$ and the commutative (associative) multiplication $\cdot_P$ are given by
\begin{eqnarray}\label{sub-adj-poisson}
x \cdot_P y=x \ast_P y+y\ast_P x,\quad\{x,y\}_{P}=x\circ_{P}y-y\circ_{P}x, \quad \forall x,y\in P.
\end{eqnarray}
\end{pro}

\emptycomment{
\begin{proof}
We only need to prove that Lie bracket $\{\cdot,\cdot\}_P$ and the commutative (associative) multiplication $\cdot_P$ satisfy the Leibniz rule. By \eqref{pre-poisson1} and \eqref{pre-poisson2}, we have
\begin{eqnarray*}
&&\{x,y\cdot_P z\}_P \\
&=&  \{x,y\ast_P z+z \ast_P y\}_P\\
&=&  x \circ_P (y \ast_P z+z \ast_P y)-(y \ast_P z+z\ast_P y)\circ_P x\\
&=& (x \circ_P y-y\circ_P x)\ast_P z+z \ast_P (x \circ_P y-y\circ_P x)+y \ast_P (x\circ_P z-z\circ_P x)+(x\circ_P z-z\circ_P x)\ast_P y\\
&=& (x \circ_P y-y\circ_P x)\cdot_P z+y\cdot_P (x \circ_P z-z \circ_P x)\\
&=& \{x,y \}_P \cdot_P z+ y\cdot_P \{x, z\}_P.
\end{eqnarray*}
Thus, $(P,\cdot_P,\{\cdot,\cdot\}_P)$ is a Poisson algebra.
\end{proof}
}

In this case, $(P,\cdot_P,\{\cdot,\cdot\}_P)$ is called the {\bf sub-adjacent Poisson algebra} of pre-Poisson algebra $(P,\ast_P,\circ_P)$, and denoted by $P^c$. $(P,\ast_P,\circ_P)$ is called a {\bf compatible pre-Poisson algebra} on the Poisson algebra $(P,\cdot_P,\{\cdot,\cdot\}_P)$.

\begin{pro}{\rm (\cite{Aguiar,LB})} \label{rep-def-pre-Poisson}
$(P,\ast_P,\circ_P)$ is a pre-Poisson algebra if and only if $(P,\cdot_P,\{\cdot,\cdot\}_P)$ given by \eqref{sub-adj-poisson} is a Poisson algebra and $(P;L_{\ast},L_{\circ})$ is a representation of $(P,\cdot_P,\{\cdot,\cdot\}_P)$.
\end{pro}

Recall that a {\bf Rota-Baxter operator of weight $\lambda$} on a Poisson algebra $(P,\cdot_P,\{\cdot,\cdot\}_P)$ is a linear map $\huaB:P \to P$ satisfying
\begin{eqnarray}
\label{eq:Poisson-RB1} \huaB(x)\cdot_P \huaB(y)&=&\huaB(\huaB(x)\cdot_P y+x \cdot_P \label{eq:Poisson-RB2} \huaB(y)+\lambda x \cdot_P y);\\
    \{\huaB(x),\huaB(y)\}_P&=&\huaB(\{\huaB(x), y\}_P+\{x, \huaB(y)\}_P+\lambda \{x,y\}_P),  \quad  \forall x,y\in P.
\end{eqnarray}

Next proposition shows that a Rota-Baxter operator of weight zero on a Poisson algebra can give rise to a pre-Poisson algebra.

\begin{pro}{\rm (\cite{Aguiar})}
Let $\huaB:P \to P$ be a Rota-Baxter operator of weight zero on a Poisson algebra $(P,\cdot_P,\{\cdot,\cdot\}_P)$. Define new operations $\ast_{\huaB}$ and $\circ_{\huaB} $ on $P$ by
$$ x\ast_{\huaB} y =\huaB(x) \cdot_P  y, \quad x\circ_{\huaB} y=\{\huaB(x),y\}_P, \quad \forall x,y\in P.$$
Then $(P,\ast_{\huaB},\circ_{\huaB})$ is a pre-Poisson algebra and $\huaB$ is a homomorphism from the sub-adjacent Poisson algebra $(P,\cdot_{\huaB},\{\cdot,\cdot\}_{\huaB})$ to $(P,\cdot_P,\{\cdot,\cdot\}_P)$, where $x \cdot_{\huaB} y =x\ast_{\huaB} y+y \ast_{\huaB} x$ and $\{x,y\}_{\huaB}=x\circ_{\huaB} y-y\circ_{\huaB} x.$
\end{pro}

Let us recall the notions of representations of Zinbiel algebras and pre-Lie algebras. Then we introduce the representations of pre-Poisson algebras.

\begin{defi}{\rm(\cite{Wang})}
A representation of a Zinbiel algebra $(A,\ast_A)$ is a triple $(V;\rho,\mu)$, where $V$ is a vector space and $\rho,\mu:A \to \gl(V)$ are linear maps such that the following equations hold for all $x,y \in A$,
\begin{eqnarray*}
&&\rho(x)\rho(y)=\rho(x\ast_A y)+\rho(y\ast_A x);\\
&&\rho(x)\mu(y)=\mu(x\ast_A y)=\mu(y)\rho(x)+\mu(y)\mu(x);
\end{eqnarray*}
\end{defi}

\begin{defi}{\rm(\cite{Bai})}
A representation of a pre-Lie algebra $(\g,\circ_{\g})$ is a triple $(V;\theta,\gamma)$, where $V$ is a vector space, $\theta:\g\longrightarrow\gl(V)$ is a representation of the sub-adjacent Lie algebra $\g^{c}$ on $V$ and $\gamma:\g\longrightarrow\gl(V)$ is a linear map satisfying
\begin{equation*}
  \theta(x)\gamma(y)-\gamma(y)\theta(x)=\gamma(x\circ_{\g}y)-\gamma(y)\gamma(x),\quad\forall x,y\in \g.
\end{equation*}
\end{defi}

\begin{defi}
A {\bf representation of a pre-Poisson algebra} $(P,\ast_P,\circ_P)$ is a quintuple $(V;\rho,\mu,\theta,\gamma)$, where $(V;\rho,\mu)$ is a representation of a Zinbiel algebra $(P,\ast_P)$ and $(V;\theta,\gamma)$ is a representation of a pre-Lie algebra $(P,\circ_P)$ such that
four linear maps $\rho,\mu,\theta,\gamma:P \to \gl(V)$ satisfy the following conditions:
\begin{eqnarray}
\label{Rep1}\rho(x\circ_P y-y\circ_P x)&=&\theta(x)\rho(y)-\rho(y)\theta(x);\\
\label{Rep2}\mu(x\circ_P y)&=&\mu(y)\gamma(x)-\mu(y)\theta(x)+\theta(x)\mu(y);\\
\label{Rep3}\mu(x\circ_P y)&=&\gamma(y)\rho(x)+\gamma(y)\mu(x)-\rho(x)\gamma(y);\\
\label{Rep4}\theta(x\ast_P y+y\ast_P x)&=&\rho(x)\theta(y)+\rho(y)\theta(x);\\
\label{Rep5}\gamma(x\ast_P y)&=& \rho(x)\gamma(y)+\mu(y)\gamma(x)-\mu(y)\theta(x),
\end{eqnarray}
for all $x,y\in P$.
\end{defi}

In fact, $(V;\rho,\mu,\theta,\gamma)$ is a representation of a pre-Poisson algebra $(P,\ast_P,\circ_P)$ if and only if the direct sum $P\oplus V$ of vector spaces is a pre-Poisson algebra with the multiplications $\ast_\ltimes$ and $\circ_\ltimes$ given by
\begin{eqnarray*}
 (x+u)\ast_\ltimes(y+v) &=&x \ast_P y+\rho(x)v+\mu(y)u;\\
 (x+u)\circ_\ltimes(y+v) &=&x \circ_P y+\theta(x)v+\gamma(y)u,
\end{eqnarray*}
for all $x,y \in P,u,v\in V.$ We denote it by $P\ltimes_{\rho,\mu,\theta,\gamma} V$ or simply by $P\ltimes V$.

\begin{ex}{\rm
Let $(P,\ast_P,\circ_P)$ be a pre-Poisson algebra. Then $(P;L_\ast,R_\ast,L_\circ,R_\circ)$ is a representation of $(P,\ast_P,\circ_P)$, which is called the {\bf regular representation}.
}
\end{ex}

If there is a pre-Poisson algebra structure on the dual space $P^*$, we denote the left and right multiplications
by $\huaL_\ast,\huaL_\circ$ and $\huaR_\ast,\huaR_\circ$  respectively.

Now we investigate the dual representation in the pre-Poisson algebra context, in order to relate pre-Poisson bialgebras and Manin triples of pre-Poisson algebras later.

\begin{pro}\label{dual-rep}
Let $(V;\rho,\mu,\theta,\gamma)$ be a representation of a pre-Poisson algebra $(P,\ast_P,\circ_P)$. Then
$$(V^{\ast};-\rho^{\ast}-\mu^{\ast},\mu^{\ast},\theta^{\ast}-\gamma^{\ast},-\gamma^{\ast})$$
is a representation of $(P,\ast_P,\circ_P)$, which is called the {\bf dual representation} of $(V;\rho,\mu,\theta,\gamma)$.
\end{pro}

\begin{proof}
Since $(V^{\ast};-\rho^{\ast}-\mu^{\ast},\mu^{\ast})$ and $(V^{\ast};\theta^{\ast}-\gamma^{\ast},-\gamma^{\ast})$ are the dual representations of $(V;\rho,\mu)$ and $(V;\theta,\gamma)$ when we consider the Zinbiel algebra $(P,\ast_P)$ and the pre-Lie algebra $(P,\circ_P)$ respectively, we only need to prove that four linear maps $-\rho^{\ast}-\mu^{\ast},\mu^{\ast},\theta^{\ast}-\gamma^{\ast},-\gamma^{\ast}$ satisfy the Equations \eqref{Rep1}-\eqref{Rep5}. For all $x,y,z\in P$ and $\xi\in P^*$, we have
\begin{eqnarray*}
&&\langle \Big((-\rho^{\ast}-\mu^{\ast})(x\circ_P y-y \circ_P x)+(\theta^*-\gamma^*)(x)(\rho^*+\mu^*)(y)-(\rho^*+\mu^*)(y)(\theta^*-\gamma^*)(x)\Big)\xi,z\rangle\\
&=& \langle \xi,(\rho+\mu)(x\circ_P y-y \circ_P x)z+(\rho+\mu)(y)(\theta-\gamma)(x)z- (\theta-\gamma)(x)(\rho+\mu)(y)z \rangle\\
&=& \langle \xi,\rho(x\circ_P y-y \circ_P x)z+\mu(x\circ_P y-y \circ_P x)z+\rho(y)\theta(x)z-\rho(y)\gamma(x)z+\mu(y)\theta(x)z\\
&&-\mu(y)\gamma(x)z-\theta(x)\rho(y)z+\gamma(x)\rho(y)z-\theta(x)\mu(y)z+\gamma(x)\mu(y)z \rangle\\
&=& 0,
\end{eqnarray*}
which implies that \eqref{Rep1} holds. Similarly, Equations \eqref{Rep2}-\eqref{Rep5} hold. Thus, $(V^{\ast};-\rho^{\ast}-\mu^{\ast},\mu^{\ast},\theta^{\ast}-\gamma^{\ast},-\gamma^{\ast})$
is a representation of $(P,\ast_P,\circ_P)$.
\end{proof}

\begin{ex}{\rm
Let $(P,\ast_P,\circ_P)$ be a pre-Poisson algebra. By Proposition \ref{dual-rep}, $(P^{\ast};-L_{\ast}^{\ast}-R_{\ast}^{\ast},R_{\ast}^{\ast},L_{\circ}^{\ast}-R_{\circ}^{\ast},
-R_{\circ}^{\ast})$ is a representation of $(P,\ast_P,\circ_P)$, which is called the {\bf coregular representation}, where four linear maps $L_{\ast}^{\ast},R_{\ast}^{\ast},L_{\circ}^{\ast},R_{\circ}^{\ast}:P\to \gl(P^{\ast})$ are respectively defined by
\begin{eqnarray*}
\langle L_{\ast}^{\ast} (x) (\xi),y \rangle&=&-\langle \xi, x\ast_P y \rangle, \quad \langle R_{\ast}^{\ast} (x) (\xi),y \rangle=-\langle \xi, y\ast_P x \rangle,\\
\langle L_{\circ}^{\ast} (x) (\xi),y \rangle&=&-\langle \xi, x\circ_P y \rangle, \quad \langle R_{\circ}^{\ast} (x) (\xi),y \rangle=-\langle \xi, y\circ_P x \rangle,\quad  \forall x,y\in A, \xi \in A^{\ast}.
\end{eqnarray*}
}
\end{ex}

Next we introduce the notions of a symplectic structure on a Poisson algebra and a phase space of a Poisson algebra.

\begin{defi}
A {\bf symplectic structure} on a Poisson algebra $(P,\cdot_P,\{\cdot,\cdot\}_P)$ is a non-degenerate skew-symmetric bilinear form $\omega\in \wedge^2 P^*$ which is a symplectic structure on the Lie algebra $(P,\{\cdot,\cdot\}_P)$ and a Connes cocycle on the commutative associative algebra $(P,\cdot_P)$, i.e.,
\begin{eqnarray}
\label{sym-1}\omega(\{x,y\}_P,z)+\omega(\{y,z\}_P,x)+\omega(\{z,x\}_P,y)&=&0,\\
\label{sym-2}\omega(x\cdot_P y,z)+\omega(y\cdot_P z,x)+\omega(z\cdot_P x,y)&=&0,
\end{eqnarray}
for all $x,y,z\in P$.
\end{defi}

\begin{pro}\label{sym-Poisson-pre-Poisson}
Let $(P,\cdot_P,\{\cdot,\cdot\}_P,\omega)$ be a symplectic Poisson algebra. Then there exists a compatible pre-Poisson algebra $(P,\ast_P,\circ_P)$ given by
\begin{eqnarray}
\label{Poisson-pre1} \omega(x\ast_P y,z)&=&\omega(y,x\cdot_P z);\\
\label{Poisson-pre2} \omega(x\circ_P y,z)&=&-\omega(y,\{x, z\}_P), \quad \forall x,y,z\in P.
\end{eqnarray}
\end{pro}

\begin{proof}
From \cite[Propositon 2.7]{LB} and \cite[Theorem 2.2]{Bai}, we can obtain that $(P,\ast_P)$ is a Zinbiel algebra and $(P,\circ_P)$ is a pre-Lie algebra. We only need to prove that $(P,\ast_P)$ and $(P,\circ_P)$ satisfy the compatibility conditions \eqref{pre-poisson1} and \eqref{pre-poisson2} respectively. For all $x,y,z,t\in P$, by \eqref{Poisson-pre1} and \eqref{Poisson-pre2}, we have
\begin{eqnarray*}
&&\omega((x\circ_P y-y\circ_P x)\ast_P z-x\circ_P (y\ast_P z)+y\circ_P (x\ast_P z),t)\\
&=& \omega(z,(x\circ_P y-y\circ_P x)\cdot_P t)+\omega(y\ast_P z,\{x,t\}_P)+\omega(x\ast_P z,y\cdot_P t)\\
&=& \omega(z,\{x,y\}_P\cdot_P t+y\cdot_P \{x,t\}_P-\{x,y\cdot_P t\})\\
&=& 0,
\end{eqnarray*}
and
\begin{eqnarray*}
&&\omega((x\ast_P y+y\ast_P x)\circ_P z-x\ast_P (y\circ_P z)-y\ast_P (x\circ_P z),t)\\
&=& -\omega(z,\{x\cdot_P y,t\})-\omega(y\circ_P z,x\cdot_P t)-\omega(x\circ_P z,y\cdot_P t)\\
&=& \omega(z,\{t,x\cdot_P y\}_P +\{x,y\cdot_P t\}_P +\{y,t\cdot_P x\}_P )\\
&=& 0,
\end{eqnarray*}
which implies that $(P,\ast_P)$ and $(P,\circ_P)$ satisfy the compatibility conditions \eqref{pre-poisson1} and \eqref{pre-poisson2}. Thus, $(P,\ast_P,\circ_P)$ is a compatible  pre-Poisson algebra of a Poisson algebra $(P,\cdot_P,\{\cdot,\cdot\}_P)$.
\end{proof}

\begin{defi}
Let $(P,\cdot_P,\{\cdot,\cdot\}_P)$ be a Poisson algebra and $P^*$ its dual space.
If there is a Poisson algebra structure on the direct sum vector space $T^*P=P\oplus P^*$ such that $(T^*P,\cdot,\{\cdot,\cdot\},\omega)$ is a symplectic Poisson algebra, where $\omega$ is given by \eqref{natural-skew-sym-form}, and $(P,\cdot_P,\{\cdot,\cdot\}_P)$ and $(P^*,\cdot|_{P^*},\{\cdot,\cdot\}|_{P^*})$ are Poisson subalgebras of $(T^*P,\cdot,\{\cdot,\cdot\})$, then the symplectic Poisson algebra $(T^*P,\cdot,\{\cdot,\cdot\},\omega)$ is called {\bf a phase space of the Poisson algebra} $(P,\cdot_P,\{\cdot,\cdot\}_P)$.
\end{defi}

Pre-Poisson algebras play an important role in the study of phase spaces of Poisson algebras.

\begin{thm}\label{phasespace-subadj}
A Poisson algebra has a phase space if and only if it is sub-adjacent to a pre-Poisson algebra.
\end{thm}

\begin{proof}
Let $(P,\ast_P,\circ_P)$ be a pre-Poisson algebra. By proposition \ref{rep-def-pre-Poisson}, $(P;L_{\ast},L_{\circ})$ is a representation of the sub-adjacent Poisson algebra $(P^c,\cdot_P,\{\cdot,\cdot\}_P)$. Then $(P^*;-L^*_{\ast},L^*_{\circ})$ is also a representation of $(P^c,\cdot_P,\{\cdot,\cdot\}_P)$. Thus, we have the semi-direct product Poisson algebra $(P^c\ltimes_{-L^*_{\ast},L^*_{\circ}} P^*,\cdot_{\ltimes},\{\cdot,\cdot\}_{\ltimes})$. By \eqref{natural-skew-sym-form}, for all $x,y,z\in P $ and $\xi,\eta,\zeta\in P^*$, we have
\begin{eqnarray*}
&&\omega(\{x+\xi,y+\eta\}_{\ltimes},z+\zeta)+\omega(\{y+\eta,z+\zeta\}_{\ltimes},x+\xi)
+\omega(\{z+\zeta,x+\xi\}_{\ltimes},y+\eta)\\
&=& \omega(\{x,y\}_P+L^*_{\circ}(x)\eta-L^*_{\circ}(y)\xi,z+\zeta )+ \omega(\{y,z\}_P+L^*_{\circ}(y)\zeta-L^*_{\circ}(z)\eta,x+\xi )\\
&&+\omega(\{z,x\}_P+L^*_{\circ}(z)\xi-L^*_{\circ}(x)\zeta,y+\eta )\\
&=& -\langle \eta,x\circ_P z \rangle+\langle \xi,y\circ_P z\rangle-\langle \theta,\{x,y\}_P \rangle-\langle \zeta,y\circ_P x \rangle+\langle \eta,z\circ_P x\rangle-\langle \xi,\{y,z\}_P \rangle-\langle \xi,z\circ_P y \rangle\\
&&+\langle \zeta,x\circ_P y\rangle-\langle \eta,\{z,x\}_P \rangle\\
&=& 0,
\end{eqnarray*}
and
\begin{eqnarray*}
&&\omega((x+\xi)\cdot_{\ltimes}(y+\eta),z+\zeta)+\omega((y+\eta)\cdot_{\ltimes}(z+\zeta),x+\xi)
+\omega((z+\zeta)\cdot_{\ltimes}(x+\xi),y+\eta)\\
&=& \omega(x\cdot_P y-L^*_*(x)\eta-L^*_*(y)\xi,z+\zeta)+\omega(y\cdot_P z-L^*_*(y)\zeta-L^*_*(z)\eta,x+\xi)\\
&&+\omega(z\cdot_P x-L^*_*(z)\xi-L^*_*(x)\zeta,y+\eta)\\
&=& \langle \eta,x\ast_P z\rangle+\langle \xi,y\ast_P z\rangle-\langle \zeta,x\cdot_P y\rangle
+\langle \zeta,y\ast_P x\rangle+\langle \eta,z\ast_P x\rangle-\langle \xi,y\cdot_P z\rangle
+\langle \xi,z\ast_P y\rangle\\
&&+\langle \zeta,x\ast_P y\rangle-\langle \eta,z\cdot_P x\rangle\\
&=& 0,
\end{eqnarray*}
which implies that $\omega$ is a symplectic structure on the semi-direct product Poisson algebra $(P^c\ltimes_{-L^*_{\ast},L^*_{\circ}} P^*,\cdot_{\ltimes},\{\cdot,\cdot\}_{\ltimes})$. Moreover, $(P^c,\cdot_P,\{\cdot,\cdot\}_P)$ is a subalgebra of $P^c\ltimes_{-L^*_{\ast},L^*_{\circ}} P^*$ and $P^*$ is a trivial Poisson subalgebra of  $P^c\ltimes_{-L^*_{\ast},L^*_{\circ}} P^*$. Thus, the symplectic Poisson algebra $(P^c\ltimes_{-L^*_{\ast},L^*_{\circ}} P^*,\cdot_{\ltimes},\{\cdot,\cdot\}_{\ltimes},\omega)$ is a phase space of the sub-adjacent Poisson algebra $(P^c,\cdot_P,\{\cdot,\cdot\}_P)$.

Conversely, let $(T^*P=P\oplus P^*,\cdot,\{\cdot,\cdot\},\omega)$  be a phase space of a Poisson algebra $(P,\cdot_P,\{\cdot,\cdot\}_P)$. By Proposition \ref{sym-Poisson-pre-Poisson}, there exists a compatible pre-Poisson algebra structure $(\ast,\circ)$ on $T^*P$ given by \eqref{Poisson-pre1} and \eqref{Poisson-pre2}. Since $(P,\cdot_P,\{
\cdot,\cdot\}_P)$ is a subalgebra of $T^*P$, we have
\begin{eqnarray*}
\omega(x\ast y,z)&=&\omega(y,x\cdot_P z)=0;\\
\omega(x\circ y,z)&=&-\omega(y,\{x, z\}_P)=0,\quad \forall x,y,z\in P.
\end{eqnarray*}
Thus, $x\ast y\in P$ and $x\circ y\in P$, which implies that $(P,\ast|_P,\circ|_P)$ is a subalgebra of the pre-Poisson algebra $(T^*P=P\oplus P^*,\ast,\circ)$. Its sub-adjacent Poisson algebra $(P^c,\cdot_c,\{\cdot,\cdot\}_c)$ is exactly the original Poisson algebra $(P,\cdot_P,\{\cdot,\cdot\}_P)$.
\end{proof}

\begin{cor}\label{phase-space-sub}
Let $(T^*P=P\oplus P^*,\cdot,\{\cdot,\cdot\},\omega)$  be a phase space of a Poisson algebra $(P,\cdot_P,\{\cdot,\cdot\}_P)$ and $(P\oplus P^*,\ast,\circ)$ the associated pre-Poisson algebra. Then both $(P,\ast|_P,\circ|_P)$ and $(P^*,\ast|_{P^*},\circ|_{P^*})$ are subalgebras of the pre-Poisson algebra $(P\oplus P^*,\ast,\circ)$.
\end{cor}

\begin{cor}
Let $(T^*P=P\oplus P^*,\cdot,\{\cdot,\cdot\},\omega)$  be a phase space of a Poisson algebra $(P,\cdot_P,\{\cdot,\cdot\}_P)$ such that the Poisson algebra $(P\oplus P^*,\cdot,\{\cdot,\cdot\})$ is a semidirect product $P\ltimes_{-\alpha^*,\beta^*} P^*$, where
$(\alpha,\beta)$ is a representation of $(P,\cdot_P,\{\cdot,\cdot\}_P)$ on $P$ and $(-\alpha^*,\beta^*)$ is its
dual representation, then
$$x\ast_P y:=\alpha(x)y,\quad x\circ_P y:=\beta(x)y,\quad \forall x,y,z\in P,$$
defines a pre-Poisson algebra structure on $P$.
\end{cor}

\begin{proof}
For all $x,y,z\in h$ and $\xi\in h^*$, by \eqref{Poisson-pre1} and \eqref{Poisson-pre2}, we have
\begin{eqnarray*}
\langle \xi,x\ast_P y\rangle&=&-\omega(x\ast_P y,\xi)=-\omega(y,x\cdot_{P\oplus P^*} \xi)
=\omega(y,\alpha^*(x)\xi)
= -\langle\alpha^*(x)\xi, y\rangle=\langle  \xi,\alpha(x)y \rangle;\\
\langle \xi,x\circ_P y\rangle&=&-\omega(x\circ_P y,\xi)=\omega(y,\{x, \xi\}_{P\oplus P^*})
=\omega(y,\beta^*(x)\xi)
= -\langle\beta^*(x)\xi, y\rangle=\langle  \xi,\beta(x)y \rangle.
\end{eqnarray*}
Therefore, we have $x\ast_P y=\alpha(x)y$ and $x\circ_P y=\beta(x)y$.
\end{proof}

At the end of this section, we introduce the notion of a Manin triple of pre-Poisson algebras.

\begin{defi}
{\bf A quadratic pre-Poisson algebra} is a pre-Poisson algebra $(P,\ast_P,\circ_P)$ equipped with a nondegenerate skew-symmetric bilinear form $\omega\in \wedge^2 P^*$ such that the following invariant conditions hold:
\begin{eqnarray}
\label{quadratic1}\omega(x\ast_P y,z)&=&\omega(y,x\ast_P z+z\ast_P x);\\
\label{quadratic2}\omega(x\circ_P y,z)&=&-\omega(y,x\circ_P z-z\circ_P x), \quad \forall x,y,z\in P.
\end{eqnarray}
\end{defi}

\begin{rmk}
Actually, by \eqref{quadratic1} and \eqref{quadratic2}, the following equalities also hold:
\begin{equation}\label{quadratic3}
 \omega(x\ast_P y,z)=\omega(z\ast_P y,x),\quad \omega(x\circ_P y,z)=-\omega(z\circ_P y,x), \quad \forall x,y,z\in P.
\end{equation}
\end{rmk}

\begin{defi}
A {\bf Manin triple of pre-Poisson algebras} is a triple $(\huaP,P,P')$, where
\begin{itemize}
\item[{\rm(i)}] $(\huaP,\ast_{\huaP},\circ_{\huaP},\omega)$ is a quadratic pre-Poisson algebra.

\item[{\rm(ii)}] both $P$ and $P'$ are isotropic subalgebras of $(\huaP,\ast_{\huaP},\circ_{\huaP})$ with respect to $\omega$.

\item[{\rm(iii)}]$\huaP=P\oplus P'$ as vector spaces.
\end{itemize}
\end{defi}

\begin{ex}
Let $(P,\ast_P,\circ_P)$ be a pre-Poisson algebra. Then $(P\ltimes_{-L_{\ast}^{\ast}-R_{\ast}^{\ast},R_{\ast}^{\ast},L_{\circ}^{\ast}-R_{\circ}^{\ast},
-R_{\circ}^{\ast}} P^*,P,P^*)$ is a Manin triple of pre-Poisson algebras, where the natural nondegenerate skew-symmetric bilinear form $\omega$ on $P\oplus P^*$ is given by \eqref{natural-skew-sym-form}.
\end{ex}

There is a one-to-one correspondence between Manin triples of pre-Poisson algebras and phase spaces of Poisson algebras.

\begin{thm}\label{phase-space-Manin-triple}
If $(P\oplus P^*,P,P^*)$ is a Manin triple of pre-Poisson algebras, then $((P \oplus P^*)^c,\cdot,\{\cdot,\cdot\},\omega)$ is a phase space of the Poisson algebra $(P^c,\cdot_P,\{\cdot,\cdot\}_P)$.

Conversely, if $(P\oplus P^*,\cdot_{P\oplus P^*},\{\cdot,\cdot\}_{P\oplus P^*},\omega)$ is a phase spaces of a Poisson algebra $(P,\cdot_P,\{\cdot,\cdot\}_P)$, then $(P\oplus P^*,P,P^*)$ is a Manin triple of pre-Poisson algebras, where the pre-Poisson algebra structure on $P\oplus P^*$ is given by \eqref{Poisson-pre1} and \eqref{Poisson-pre2}.
\end{thm}

\begin{proof}
Let $(P\oplus P^*,P,P^*)$ be a Manin triple of pre-Poisson algebras. Denote by $(\ast_P,\circ_P)$ and $(\ast_{P^*},\circ_{P^*})$ the pre-Poisson algebra structure on $P$ and $P^*$ respectively, and denote by $(\cdot_P,\{\cdot,\cdot\}_P)$ and $(\cdot_{P^*},\{\cdot,\cdot\}_{P^*})$ the corresponding sub-adjacent Poisson structure on $P^c$ and $(P^*)^c$ respectively. It is straightforward to deduce that the corresponding Poisson algebra structure $(\cdot,\{\cdot,\cdot\})$ on the $(P \oplus P^*)^c$ is given by
\begin{eqnarray*}
(x+\xi)\cdot(y+\eta)&=&x\cdot_P y-\huaL^*_*(\xi)y-\huaL^*_*(\eta)x+\xi\cdot_{P^*} \eta-L^*_*(x)\eta-L^*_*(y)\xi;\\
\{x+\xi,y+\eta\}&=&\{x, y\}_P+\huaL^*_{\circ}(\xi)y-\huaL^*_{\circ}(\eta)x+\{\xi,\eta\}_{P^*} +L^*_{\circ}(x)\eta-L^*_{\circ}(y)\xi.
\end{eqnarray*}
For all $x,y,z\in P$ and $\xi,\eta,\zeta\in P^*$, we have
\begin{eqnarray*}
&&\omega(\{x+\xi,y+\eta\},z+\zeta)+\omega(\{y+\eta,z+\zeta\},x+\xi)+\omega(\{z+\zeta, x+\xi\},y+\eta)\\
&=&\omega(\{x, y\}_P+\huaL^*_{\circ}(\xi)y-\huaL^*_{\circ}(\eta)x+\{\xi,\eta\}_{P^*} +L^*_{\circ}(x)\eta-L^*_{\circ}(y)\xi,z+\zeta)
+\omega(\{y, z\}_P+\huaL^*_{\circ}(\eta)z\\
&&-\huaL^*_{\circ}(\zeta)y+\{\eta,\zeta\}_{P^*} +L^*_{\circ}(y)\zeta-L^*_{\circ}(z)\eta,x+\xi)
+\omega(\{z, x\}_P+\huaL^*_{\circ}(\zeta)x-\huaL^*_{\circ}(\xi)z+\{\zeta,\xi\}_{P^*}\\ &&+L^*_{\circ}(z)\xi-L^*_{\circ}(x)\zeta,y+\eta)\\
&=& \langle \{\xi,\eta\}_{P^*} +L^*_{\circ}(x)\eta-L^*_{\circ}(y)\xi,z \rangle-\langle\zeta, \{x, y\}_P+\huaL^*_{\circ}(\xi)y-\huaL^*_{\circ}(\eta)x \rangle
+\langle \{\eta,\zeta\}_{P^*} +L^*_{\circ}(y)\zeta\\
&&-L^*_{\circ}(z)\eta,x\rangle-\langle \xi,\{y, z\}_P+\huaL^*_{\circ}(\eta)z-\huaL^*_{\circ}(\zeta)y \rangle
+\langle \{\zeta,\xi\}_{P^*} +L^*_{\circ}(z)\xi-L^*_{\circ}(x)\zeta,y \rangle\\
&&-\langle \eta,\{z, x\}_P+\huaL^*_{\circ}(\zeta)x-\huaL^*_{\circ}(\xi)z \rangle\\
&=& \langle \{\xi,\eta\}_{P^*},z \rangle-\langle \eta, x\circ_P z\rangle+\langle \xi, y\circ_P z \rangle-\langle \zeta, \{x,y\}_P \rangle+\langle \xi\circ_{P^*} \zeta, y\rangle-\langle\xi \circ_{P^*} \eta, z\rangle+\langle \{\eta,\zeta\}_{P^*},x \rangle\\
&&-\langle \zeta, y\circ_P x\rangle+\langle \eta, z\circ_P x \rangle-\langle \xi, \{y,z\}_P \rangle+\langle \eta \circ_{P^*} \xi, z\rangle-\langle\eta \circ_{P^*} \zeta, x\rangle
+\langle \{\zeta,\xi \}_{P^*},y \rangle\\
&&-\langle \xi, z\circ_P y\rangle+\langle \zeta, x\circ_P y \rangle-\langle \eta, \{z,x\}_P \rangle+\langle \zeta \circ_{P^*} \eta, x\rangle-\langle \zeta \circ_{P^*} \xi, y\rangle\\
&=& 0.
\end{eqnarray*}
Similarly, by a direct calculation, we also have
$$\omega((x+\xi)\cdot(y+\eta),z+\zeta)+\omega((y+\eta)\cdot(z+\zeta),x+\xi)+\omega((z+\zeta)\cdot( x+\xi),y+\eta)=0.$$
Thus, we deduce that $\omega$ is a symplectic structure on the Poisson algebra $((P \oplus P^*)^c,\cdot,\{\cdot,\cdot\})$. Therefore, it is a phase space.

Conversely, let $(P\oplus P^*,\cdot_{P\oplus P^*},\{\cdot,\cdot\}_{P\oplus P^*},\omega)$ be a phase space of a Poisson algebra $(P,\cdot_P,\{\cdot,\cdot\}_P)$. By Proposition \ref{sym-Poisson-pre-Poisson}, there exists a pre-Poisson algebra $(\ast,\circ)$ on $P\oplus P^*$ given by \eqref{Poisson-pre1} and \eqref{Poisson-pre2} such that $(P\oplus P^*,\ast,\circ,\omega)$ is a quadratic pre-Poisson algebra. By Corollary \ref{phase-space-sub},
$(P,\ast|_P,\circ|_P)$ and $(P^*,\ast|_{P^*},\circ|_{P^*})$ are subalgebras of the pre-Poisson algebra $(P\oplus P^*,\ast,\circ)$. It is obvious that both $P$ and $P^*$ are isotropic. Therefore, $(P\oplus P^*,P,P^*)$ is a Manin triple of pre-Poisson algebras.
\end{proof}

\section{Pre-Poisson bialgebras and the pre-Poisson Yang-Baxter equation}\label{sec:pre-Poisson-bialgebras}

In this section, we first introduce the notion of pre-Poisson bialgebras. Then we prove that pre-Poisson bialgebras, Manin triples of pre-Poisson algebras and phase spaces of the sub-adjacent Poisson algebras are equivalent. Moreover, we introduce the notion of coboundary pre-Poisson bialgebras, which leads to an analogue of the classical Yang-Baxter equation.

\subsection{Pre-Poisson bialgebras}
Let us first introduce the notion of pre-Poisson bialgebras.

\begin{defi}
A {\bf pre-Poisson bialgebra} is a sextuple $(P,P^*,\ast_P,\circ_P,\Delta,\delta)$, where $(P,\ast_P,\circ_P)$ is a pre-Poisson algebra and $\Delta,\delta:P\to P \otimes P$ are two linear maps such that
\begin{itemize}
\item[{\rm(i)}] $(P^*,\Delta^*,\delta^*)$ is a pre-Poisson algebra, where $\Delta^*,\delta^*:P^* \otimes P^* \to P^* $ are the dual maps of $\Delta$ and $\delta$ respectively. That means, $(P^*,\Delta^*)$ is a Zinbiel algebra and $(P^*,\delta^*)$ is a pre-Lie algebra such that the following compatible conditions hold:
    \begin{eqnarray}
    \label{Poisson-bialg-5} (\delta\otimes \Id)\Delta-(\tau\otimes \Id)(\delta\otimes \Id)\Delta-(\Id \otimes \Delta)\delta+(\tau\otimes \Id)(\Id \otimes \delta)\Delta&=&0;\\
    \label{Poisson-bialg-6} (\Delta\otimes \Id)\delta+(\tau\otimes \Id)(\Delta\otimes \Id)\delta-(\Id \otimes \delta)\Delta-(\tau\otimes \Id)(\Id \otimes \delta)\Delta&=&0.
    \end{eqnarray}

\item[{\rm(ii)}] $(P,P^*,\ast_P,\Delta)$ is a Zinbiel bialgebra and $(P,P^*,\circ_P,\delta)$ is a pre-Lie bialgebra.

\item[{\rm(iii)}] For all $x,y\in P$, $\Delta$ and $\delta$ satisfy the following compatibility conditions:
  \begin{eqnarray}
  \label{Poisson-bialg-1}\delta(x\ast_P y+y\ast_P x)&=&(\Id \otimes (L_\ast+R_\ast)(y))\delta(x)+(\Id \otimes (L_\ast+R_\ast)(x))\delta(y)\\
  \nonumber&&-(L_\circ(x)\otimes \Id)\Delta(y)-(L_\circ(y)\otimes \Id)\Delta(x);\\
  \label{Poisson-bialg-2}\Delta(x\circ_P y-y \circ_P x)&=& (L_\ast(y)\otimes \Id-\Id \otimes (L_\ast+R_\ast)(y))\delta(x)\\
  \nonumber&&+(L_\circ(x)\otimes \Id+\Id \otimes (L_\circ-R_\circ)(x))\Delta(y);\\
  \label{Poisson-bialg-3}(\Delta+\tau\Delta)(x\circ_P y)&=&-(\Id \otimes R_\ast(y))\delta(x)-(R_\ast(y)\otimes \Id)\tau\delta(x)\\
  \nonumber&&+(L_\circ(x)\otimes \Id+\Id\otimes L_\circ(x))(\Delta+\tau\Delta)(y);\\
  \label{Poisson-bialg-4}(\delta-\tau\delta)(x\ast_P y)&=&(\Id \otimes R_\ast(y))\delta(x)+(\Id \otimes L_\ast(x))(\delta-\tau\delta)(y)+(R_\circ(y)\otimes \Id)\tau\Delta(x)\\
  \nonumber&&-(L_\circ(x)\otimes \Id)(\Delta+\tau\Delta)(y).
  \end{eqnarray}
\end{itemize}
\end{defi}

\begin{ex}\label{dual pre-Poisson alg}
{\rm
Let $(P,P^*,\ast_P,\circ_P,\Delta,\delta)$ be a pre-Poisson bialgebra. Then its dual $(P^*,P,\Delta^*,\delta^*,\ast_P^*,\circ_P^*)$ is also a pre-Poisson bialgebra.
}
\end{ex}

The following theorem is the main result in this section.
\begin{thm}\label{pre-Poisson-bialgs}
Let $(P,\ast_P,\circ_P)$ and $(P^*,\ast_{P^*},\circ_{P^*})$ be two pre-Poisson algebras. Then the following sentences are equivalent:
\begin{itemize}
\item[{\rm(i)}] $(P,P^*,\ast_P,\circ_P,\Delta,\delta)$ is a pre-Poisson bialgebra, where $\Delta^*:=\ast_{P^*}$ and $\delta^*:=\circ_{P^*}$.

\item[{\rm(ii)}] $(P\oplus P^*,P,P^*)$ is a Manin triple of pre-Poisson algebras, where the skew-symmetric bilinear form $\omega$ on $P\oplus P^*$ is given by \eqref{natural-skew-sym-form}.

\item[{\rm(iii)}] $(P^c \oplus (P^*)^c,\cdot,\{\cdot,\cdot\},\omega)$ is a phase spaces of the sub-adjacent Poisson algebra $(P^c,\cdot_P,\{\cdot,\cdot\}_P)$.
\end{itemize}
\end{thm}

\begin{proof}
{\rm(ii)} is equivalent to {\rm(iii)} according to the Theorem \ref{phase-space-Manin-triple}.
Now we prove that {\rm(i)} is equivalent to {\rm(ii)}. If $(P,P^*,\ast_P,\circ_P,\Delta,\delta)$ is a pre-Poisson bialgebra, by a direct calculation using \eqref{Poisson-bialg-1}-\eqref{Poisson-bialg-4}, we deduce that there exists a pre-Poisson algebra structure $(\ast_{\bowtie},\circ_{\bowtie})$ on $P\oplus P^*$ given by
\begin{eqnarray*}
(x+\xi) \ast_{\bowtie} (y+\eta)&=& x \ast_P y-(\huaL^*_*+\huaR^*_*)(\xi)y+ \huaR^*_*(\eta)x+u \ast_{P^*} v-(L^*_*+R^*_*)(x)\eta+ R^*_*(y)\xi;\\
(x+\xi) \circ_{\bowtie} (y+\eta)&=& x \circ_P y+(\huaL^*_{\circ}-\huaR^*_{\circ})(\xi)y- \huaR^*_{\circ}(\eta)x+u \ast_{P^*} v+(L^*_{\circ}-R^*_{\circ})(x)\eta- R^*_{\circ}(y)\xi,
\end{eqnarray*}
for all $x,y\in P,\xi, \eta\in P^*.$ In this case, it is obvious that both $P$ and $P^*$ are isotropic subalgebras of $(P\oplus P^*,\ast_{\bowtie},\circ_{\bowtie})$ with respect to $\omega$.
We only need to prove that $\omega$ satisfies the invariant conditions \eqref{quadratic1} and \eqref{quadratic2}. For all $x,y,z\in P$ and $\xi,\eta,\zeta\in P^*$, we have
\begin{eqnarray*}
&&\omega((x+\xi)\ast_{\bowtie}(y+\eta),z+\zeta)\\
&=&\langle \xi\ast_{P^*} \eta-(L^*_*+R^*_*)(x)\eta+R^*_*(y)\xi,z\rangle-\langle \zeta,x\ast_P y-
(\huaL^*_*+\huaR^*_*)(\xi)y+ \huaR^*_*(\eta)x \rangle\\
&=& \langle \xi \ast_{P^*} \eta, z\rangle+\langle \eta,x\cdot_P z\rangle-\langle  \xi,z\ast_P y \rangle-\langle  \zeta,x\ast_P y \rangle-\langle \xi\cdot_{P^*} \zeta,y  \rangle+\langle \zeta\ast_{P^*} \eta,x \rangle\\
&=& \langle \eta,x\cdot_P z-\huaL^*_* (\xi)z-\huaL^*_*(\zeta) x \rangle-\langle \xi\cdot_{P^*}\zeta-L^*_*(x)\zeta-L^*_*(z)\xi,y \rangle\\
&=& \omega(y+\eta,x\cdot_P z-\huaL^*_*(\xi)z-\huaL^*_*(\zeta)x+\xi\cdot_{P^*} \zeta-L^*_*(x)\zeta-L^*_*(z)\xi)\\
&=& \omega(y+\eta,(x+\xi)\ast_{\bowtie} (z+\zeta)+(z+\zeta)\ast_{\bowtie} (x+\xi)),
\end{eqnarray*}
which implies that $\omega$ satisfies the invariant condition \eqref{quadratic1}. Similarly, we can also have
$$\omega((x+\xi)\circ_{\bowtie}(y+\eta),z+\zeta)=-\omega(y+\eta,\{x+\xi,z+\zeta\}_{\bowtie}), $$
which implies that $\omega$ satisfies the invariant condition \eqref{quadratic2}. Thus, $(P\oplus P^*,P,P^*)$ is a Manin triple of pre-Poisson algebras.

On the other hand, if $(P\oplus P^*,P,P^*)$ is a Manin triple of pre-Poisson algebras with the skew-symmetric bilinear form $\omega$ given by \eqref{natural-skew-sym-form}, by the invariant condition of $\omega$, we can deduce that the pre-Poisson algebra structure on $P\oplus P^*$ given by
\begin{eqnarray*}
(x+\xi) \ast (y+\eta)&=& x \ast_P y-(\huaL^*_*+\huaR^*_*)(\xi)y+ \huaR^*_*(\eta)x+u \ast_{P^*} v-(L^*_*+R^*_*)(x)\eta+ R^*_*(y)\xi;\\
(x+\xi) \circ (y+\eta)&=& x \circ_P y+(\huaL^*_{\circ}-\huaR^*_{\circ})(\xi)y- \huaR^*_{\circ}(\eta)x+u \circ_{P^*} v+(L^*_{\circ}-R^*_{\circ})(x)\eta- R^*_{\circ}(y)\xi,
\end{eqnarray*}
which implies that $(P,P^*,\ast_P,\circ_P,\Delta,\delta)$ is a pre-Poisson bialgebra with $\Delta^*=\ast_{P^*}$ and $\delta^*=\circ_{P^*}$ by a direct calculation. The proof is finished.
\end{proof}

\begin{defi}
Let $(P,P^*,\ast_P,\circ_P,\Delta_P,\delta_P)$ and $(Q,Q^*,\ast_Q,\circ_Q,\Delta_Q,\delta_Q)$ be two pre-Poisson bialgebras. A linear map $\varphi:P \to Q$ is called {\bf a homomorphism of pre-Poisson bialgebras}, if $\varphi:P  \to Q$ is a homomorphism of pre-Poisson algebras such that $\varphi^*:Q^* \to P^*$ is also a homomorphism of pre-Poisson algebras, i.e.,
\begin{eqnarray}
(\varphi \otimes \varphi)\circ \Delta_P=\Delta_Q \circ \varphi, \quad (\varphi \otimes \varphi)\circ \delta_P=\delta_Q \circ \varphi.
\end{eqnarray}
Moreover, if $\varphi:P \to Q$ is a linear isomorphism of vector spaces, then $\varphi$ is called {\bf an isomorphism of pre-Poisson bialgebras.}
\end{defi}

\subsection{Coboundary pre-Poisson bialgebras and the pre-Poisson Yang-Baxter equation}

In this subsection, we introduce the notion of coboundary pre-Poisson bialgebras, which gives rise to an analogue of the classical Yang-Baxter equation, named the pre-Poisson Yang-Baxter equation.

\begin{defi}
A pre-Poisson bialgebra $(P,P^*,\ast_P,\circ_P,\Delta_P,\delta_P)$ is called {\bf coboundary} if there exists an element $r\in P\otimes P$ such that
\begin{eqnarray}
\label{coboundary-condition1}\Delta(x)&=&\Big( \Id\otimes (L_\ast+R_\ast)(x)-L_\ast(x) \otimes \Id  \Big)r;\\
\label{coboundary-condition2}\delta(x)&=&\Big( L_{\circ}(x) \otimes \Id +\Id\otimes (L_{\circ}-R_{\circ})(x)  \Big)r, \quad \forall x\in P.
\end{eqnarray}
\end{defi}

Let us recall some results of coboundary Zinbiel bialgebras and coboundary pre-Lie bialgebras.

Let $(A,\ast_A)$ be a Zinbiel algebra and $r\in A\otimes A$. The linear map $\Delta$ defined by \eqref{coboundary-condition1} induces a Zinbiel algebra on $A^*$ such that $(A,A^*,\ast_{A},\Delta)$ is a coboundary Zinbiel bialgebra if and only if the following conditions are satisfied:
\begin{eqnarray}
\label{coboundary-Zinbiel-1}\Big( L_{\ast}(x\ast_A y)\otimes \Id-\Id \otimes L_{\ast} (x\ast_A y)-L_{\ast} (x) L_{\ast}(y) \otimes \Id+L_{\ast}(x)\otimes L_{\ast} (y) \Big)(r-\tau(r))&=&0;\\
\label{coboundary-Zinbiel-2}\Big(L_{\ast} (x)\otimes \Id\otimes \Id-\Id \otimes \Id \otimes (L_{\ast}+R_{\ast})(x) \Big)Z(r)&=&0,
\end{eqnarray}
where
$Z(r)=-r_{13}\ast r_{12}-r_{23}\ast r_{21}+r_{13}\cdot r_{21}+r_{12}\cdot r_{23}-r_{13}\cdot r_{23}$ is called the Zinbiel Yang-Baxter equation. See \cite{Wang} for more details.

Let $(\g,\circ_\g)$ be a pre-Lie algebra and $r\in \g\otimes \g.$ The linear map $\delta$ defined by \eqref{coboundary-condition2} induces a pre-Lie algebra on $\g^*$ such that $(\g,\g^*,\circ_{\g},\delta)$ is a coboundary pre-Lie bialgebra if and only if the following conditions are satisfied:
\begin{eqnarray}
\label{coboundary-pre-Lie-1}\Big(L_\circ (x\circ_\g y)\otimes \Id+\Id\otimes L_\circ (x\circ_\g y)-L_{\circ} (x)L_{\circ} (y)\otimes \Id -L_{\circ} (x)\otimes L_{\circ} (y)&&\\
\nonumber-L_{\circ} (y)\otimes L_{\circ} (x)-\Id\otimes L_{\circ} (x)L_{\circ} (y) \Big)(r-\tau(r))&=&0;\\
\label{coboundary-pre-Lie-2}\Big (L_\circ (x)\otimes \Id\otimes \Id+\Id\otimes  L_\circ(x)\otimes \Id+\Id\otimes \Id \otimes (L_{\circ}-R_{\circ})(x) \Big)S(r)&=&0,
\end{eqnarray}
where
$S(r)=r_{13}\circ r_{12}-r_{23}\circ r_{21}+[r_{23},r_{21}]-[r_{13},r_{21}]-[r_{13},r_{23}]$ is called the S-equation. More details can be seen in \cite{Bai}.

\begin{thm}\label{equ-coboundary-pre-Poisson-bialg}
Let $(P,\ast_P,\circ_P)$ be a pre-Poisson algebra and $r \in P \otimes P.$ Define $\Delta,\delta:P\to P\otimes P$ by \eqref{coboundary-condition1} and \eqref{coboundary-condition2} respectively. Then $\Delta,\delta$ induce a pre-Poisson algebra on $P^*$ such that $(P,P^*,\ast_P,\circ_P,\Delta,\delta)$ is a coboundary pre-Poisson bialgebra if and only if the Conditions \eqref{coboundary-Zinbiel-1}-\eqref{coboundary-pre-Lie-2} hold and the following conditions are satisfied for all $x,y \in P$:
\begin{itemize}
\item[{\rm(i)}] for $r=\sum_{i} a_i\otimes b_i \in P\otimes P,$

\begin{eqnarray*}
 &&\Big(L_\circ(x)\otimes \Id \otimes \Id \Big)Z(r)+\Big(\Id\otimes \Id \otimes (L_\ast+R_\ast)(x)-\Id\otimes L_\ast(x) \otimes \Id \Big )S(r)+\sum_{i} \Big(L_\circ(x)L_\ast(a_i)\otimes \Id\\
 &&-L_{\circ}(x \ast_P a_i)\otimes \Id-L_\circ(x)\otimes L_\ast(a_i)+L_\circ(a_i)\otimes L_\ast(x)-\Id \otimes L_{\circ}(x \ast_P a_i )+\Id \otimes L_\ast(x)L_{\circ}(a_i) \Big)\\
 &&(r-\tau(r))\otimes b_i=0;
\end{eqnarray*}

\item[{\rm(ii)}] for $r=\sum_{i} a_i\otimes b_i \in P\otimes P,$

\begin{eqnarray*}
 &&\Big( \Id\otimes \Id \otimes (L_\circ-R_\circ)(x) \Big) Z(r)+\Big( L_\ast(x)\otimes \Id\otimes \Id+\Id \otimes L_\ast(x)\otimes  \Id \Big)S(r)+\sum_{i} \Big(L_\ast(x)L_\circ(a_i)\otimes \Id\\
 &&-L_{\ast}(x \circ_P a_i)\otimes \Id-L_\ast(x)\otimes L_\circ(a_i)-L_\circ(a_i)\otimes L_\ast(x)+\Id \otimes L_{\ast}(x \circ_P a_i)-\Id \otimes L_\ast(x)L_{\circ}(a_i) \Big)\\
 &&(r-\tau(r))\otimes b_i=0;
\end{eqnarray*}

\item[{\rm(iii)}]
$\Big( L_*(y)\otimes L_{\circ}(x)-L_{\circ}(x) \otimes L_*(y)+\Id \otimes L_*(x\circ_P y)-L_*(x\circ_P y)\otimes \Id+L_{\circ}(x) L_*(y)\otimes \Id-\Id \otimes L_{\circ}(x) L_{\ast}(y) \Big) (r-\tau(r))=0;$

\item[{\rm(iv)}]
$\Big( L_{\circ}(x)\otimes L_*(y)-L_{\circ}(y)\otimes L_*(x)+L_{\circ}(x\circ_P y)\otimes \Id-
L_{\circ}(x)L_*(y)\otimes \Id+\Id \otimes L_{\circ}(x\ast_P y)-\Id\otimes L_*(x)L_{\circ}(y)
\Big) (r-\tau(r))=0.$
\end{itemize}
We denote this coboundary pre-Poisson bialgebra by $(P,P_r^*)$.
\end{thm}

\begin{proof}
First by Conditions \eqref{coboundary-Zinbiel-1} and \eqref{coboundary-Zinbiel-2}, the dual map $\Delta:\otimes^2 P^*\to P^*$ defined by \eqref{coboundary-condition1} induces a Zinbiel algebra on $P^*$ such that $(P,P^*,\ast_{P},\Delta)$ is a Zinbiel bialgebra. By by Conditions \eqref{coboundary-pre-Lie-1} and \eqref{coboundary-pre-Lie-2}, the dual map $\delta:\otimes^2 P^*\to P^*$ defined by \eqref{coboundary-condition2} induces a pre-Lie algebra on $P^*$ such that $(P,P^*,\circ_{P},\delta)$ is a pre-Lie bialgebra.
Then the Conditions {\rm(i)} and {\rm(ii)} are equivalent to \eqref{Poisson-bialg-5} and \eqref{Poisson-bialg-6} respectively, which implies that $(P^*,\Delta^*,\delta^*)$ is a pre-Poisson algebra.

Finally, by a direct calculation, if $\Delta$ and $\delta$ are given by \eqref{coboundary-condition1} and \eqref{coboundary-condition2} for some $r\in P\otimes P$, $\Delta$ and $\delta$ satisfy the compatible conditions \eqref{Poisson-bialg-1}-\eqref{Poisson-bialg-2} naturally. Moreover, the Conditions {\rm(iii)} and {\rm(iv)} are equivalent to \eqref{Poisson-bialg-3} and \eqref{Poisson-bialg-4} respectively.
\end{proof}

\begin{defi}
Let $(P,\ast_P,\circ_P)$ be a pre-Poisson algebra and $r\in P\otimes P$. The equation
\begin{equation}
Z(r)=S(r)=0
\end{equation}
is called the {\bf pre-Poisson Yang-Baxter equation} on $P$.
\end{defi}

Later on, we will use the pre-Poisson Yang-Baxter equation and some invariant conditions to define the quasi-triangular pre-Poisson bialgebras.

\section{Quasi-triangular and factorizable pre-Poisson bialgebras}\label{sec:quasi-fac}

In this section, we introduce the notions of quasi-triangular pre-Poisson bialgebras and factorizable pre-Poisson bialgebras. We show that quasi-triangular pre-Poisson bialgebras can give rise to relative Rota-Baxter operators of weight $1$. Moreover, we show that the double of a pre-Poisson bialgebra naturally enjoys a factorizable pre-Poisson bialgebra structure.

\subsection{Quasi-triangular pre-Poisson bialgebras and relative Rota-Baxter operators}

For $r\in P\otimes P,$ we define $r_+,r_-:P^*\to P$ by
\begin{eqnarray}
\langle r_+(\xi), \eta\rangle=r(\xi,\eta)=\langle \xi, r_-(\eta) \rangle, \quad \forall \xi,\eta\in P^*.
\end{eqnarray}
Then the pre-Poisson algebra structure on $P^*$ in Theorem \ref{equ-coboundary-pre-Poisson-bialg} is given by
\begin{eqnarray}
\label{Pr-1}\xi \ast_r \eta&=&-(L^*_* (r_+\xi)+R^*_* (r_+\xi))\eta+R^*_*(r_{-}\eta) \xi,\\
\label{Pr-2}\xi \circ_r \eta&=& (L^*_{\circ}-R^*_{\circ})(r_{+}\xi) \eta-R^*_{\circ}(r_{-}\eta) \xi,
\quad \forall \xi,\eta\in P^*.
\end{eqnarray}

Now we introduce the notion of  $(L,R)$-invariance of a 2-tensor $r\in P\otimes P$, which is the main ingredient in the definition of a quasi-triangular pre-Poisson bialgebra.
\begin{defi}
Let $(P,\ast_P,\circ_P)$ be a pre-Poisson algebra and $r\in P\otimes P.$ Then $r$ is  {\bf $(L,R)$-invariant} if
\begin{eqnarray}
\Big(L_{\ast}(x) \otimes \Id-\Id\otimes (L_\ast+R_\ast) (x) \Big)r&=&0,\\
\Big(L_{\circ}(x) \otimes \Id+\Id\otimes (L_{\circ}-R_{\circ})(x) \Big)r&=&0,\quad \forall x \in P.
\end{eqnarray}
\end{defi}

\begin{lem}\label{invariance1}
Let $(P,\ast_P,\circ_P)$ be a pre-Poisson algebra and $r\in P\otimes P.$ Then $r$ is $(L,R)$-invariant if and only if
\begin{eqnarray}
r_+(L_{\ast}^* (x) \xi)+x \cdot_P r_+(\xi)&=&0,\\
r_+(L_{\circ}^* (x) \xi)-\{x, r_+(\xi)\}_P&=&0,\quad \forall x\in P,\xi \in P^{\ast}.
\end{eqnarray}
\end{lem}

\begin{proof}
It follows from \cite[Lemma 4.2]{Wang} and \cite[Lemma 2.9]{WBLS}.
\end{proof}

\begin{pro}\label{pre-Poisson-inv}
Let $(P,\ast_P,\circ_P)$ be a pre-Poisson algebra and $r\in P \otimes P.$ If the skew-symmetric part $a$ of $r$ is $(L,R)$-invariant, then for all $x,y\in P$, we have
\begin{eqnarray*}
&&\Big(L_\circ(x)L_\ast(y)\otimes \Id
-L_{\circ}(x \ast_P y)\otimes \Id-L_\circ(x)\otimes L_\ast(y)+L_\circ(y)\otimes L_\ast(x)-\Id \otimes L_{\circ}(x \ast_P y )\\
&&+\Id \otimes L_\ast(x)L_{\circ}(y) \Big)a=0;\\
&&\Big(L_\ast(x)L_\circ(y)\otimes \Id
-L_{\ast}(x \circ_P y)\otimes \Id-L_\ast(x)\otimes L_\circ(y)-L_\circ(y)\otimes L_\ast(x)+\Id \otimes L_{\ast}(x \circ_P y)\\
&&-\Id \otimes L_\ast(x)L_{\circ}(y) \Big)a=0;\\
&&\Big( L_*(y)\otimes L_{\circ}(x)-L_{\circ}(x) \otimes L_*(y)+\Id \otimes L_*(x\circ_P y)-L_*(x\circ_P y)\otimes \Id+L_{\circ}(x) L_*(y)\otimes \Id\\
&&-\Id \otimes L_{\circ}(x) L_{\ast}(y) \Big) (r-\tau(r))=0;\\
&&\Big( L_{\circ}(x)\otimes L_*(y)-L_{\circ}(y)\otimes L_*(x)+L_{\circ}(x\circ_P y)\otimes \Id-
L_{\circ}(x)L_*(y)\otimes \Id+\Id \otimes L_{\circ}(x\ast_P y)\\
&&-\Id\otimes L_*(x)L_{\circ}(y)\Big) (r-\tau(r))=0.
\end{eqnarray*}
\end{pro}

\begin{proof}
Since the skew-symmetric part $a$ of $r$ is $(L,R)$-invariant, by Lemma \ref{invariance1}, for all $x,y\in P$ and $\xi,\eta\in P^*$, we have
\begin{eqnarray*}
&&\Big(L_\circ(x)L_\ast(y)\otimes \Id
-L_{\circ}(x \ast_P y)\otimes \Id-L_\circ(x)\otimes L_\ast(y)+L_\circ(y)\otimes L_\ast(x)-\Id \otimes L_{\circ}(x \ast_P y )\\
&&+\Id \otimes L_\ast(x)L_{\circ}(y) \Big)a(\xi,\eta)\\
&=& a(L^*_*(y)L^*_{\circ}(x)\xi,\eta)+ a(L^*_{\circ}(x\ast_P y)\xi,\eta)-a(L^*_{\circ}(x)\xi,L^*_*(y)\eta)+a(L^*_{\circ}(y)\xi,L^*_*(x)\eta)\\
&&+a(\xi,L^*_{\circ}(x\ast_P y)\eta)+a(\xi,L^*_{\circ}(y)L^*_*(x)\eta)\\
&=& \langle a_+(L^*_*(y)L^*_{\circ}(x)\xi),\eta\rangle+\langle a_+(L^*_{\circ}(x\ast_P y)\xi),\eta \rangle-\langle a_+(L^*_{\circ}(x)\xi), L^*_*(y)\eta  \rangle+\langle  a_+(L^*_{\circ}(y)\xi),L^*_*(x)\eta   \rangle\\
&&+\langle a_+(\xi),L^*_{\circ}
(x\ast_P y)\eta  \rangle+\langle a_+(\xi),L^*_{\circ}(y)L^*_*(x)\eta   \rangle\\
&=& \langle  -y\cdot_{P} \{x,a_+(\xi)\}_P+\{x\ast_P y,a_+(\xi)\}_P+y\ast_P \{x,a_+(\xi)\}_P-x\ast_P\{y,a_+(\xi)\}_P-(x\ast_P y)\circ_P a_+(\xi)\\
&&+x\ast_P (y\circ_P a_+(\xi)),\eta       \rangle\\
&=& \langle -\{x,a_+(\xi)\}_P \ast_P y-a_+(\xi)\circ_P (x\ast_P y)+x \ast_P (a_+(\xi) \circ_P y),\eta \rangle\\
&=& 0,
\end{eqnarray*}
which implies that
\begin{eqnarray*}
&&\Big(L_\circ(x)L_\ast(y)\otimes \Id
-L_{\circ}(x \ast_P y)\otimes \Id-L_\circ(x)\otimes L_\ast(y)+L_\circ(y)\otimes L_\ast(x)-\Id \otimes L_{\circ}(x \ast_P y )\\
&&+\Id \otimes L_\ast(x)L_{\circ}(y) \Big)a=0.
\end{eqnarray*}
We can also prove the other three equations using the same method. We omit the details. The proof is finished.
\end{proof}

\begin{defi}
Let $(P,\ast_P,\circ_P)$ be a pre-Poisson algebra. If $r\in P\otimes P$ satisfies $Z(r)=S(r)=0$ and the skew-symmetric part $a$ of $r$ is $(L,R)$-invariant, then the pre-Poisson bialgebra $(P,P_{r}^{\ast})$ induced by $r$ is called a {\bf quasi-triangular pre-Poisson bialgebra}.
Moreover, if $r$ is symmetric, then $(P,P_{r}^{\ast})$ is called a {\bf triangular pre-Poisson bialgebra}.
\end{defi}

Denote by $I$ the operator
\begin{eqnarray}\label{I}
I=r_{+}-r_{-}:P^{\ast}\longrightarrow P.
\end{eqnarray}

Note that $I^{\ast}=-I$. Actually, $\frac{1}{2} I$ is the contraction of the skew-symmetric part $a$ of $r$, which means that $\frac{1}{2} \langle I(\xi),\eta \rangle=a(\xi,\eta).  $
If $r$ is symmetric, then $I=0$.

Now we give another characterization of the $(L,R)$-invariant condition.

\begin{pro}\label{invariance2}
Let $(P,\ast_P,\circ_P)$ be a pre-Poisson algebra and $r\in P\otimes
P.$ The skew-symmetric part $a$ of $r$ is $(L,R)$-invariant if
and only if $I\circ L^*_*(x)=-(L_*(x)+R_*(x))\circ I$ and $I \circ L^*_{\circ}(x)=(L_\circ(x)-R_\circ(x))\circ I$,   for all $x\in P$,  where $I$ is given by \eqref{I}.
\end{pro}

\begin{proof}
It follows from \cite[Proposition 4.5]{Wang} and \cite[Proposition 2.12]{WBLS}.
\end{proof}

\begin{lem}
Let $(P,\ast_P,\circ_P)$ be a pre-Poisson algebra and $r\in P\otimes P.$
Let $r=a+\Lambda$ with $a\in\wedge^2 A$ and $\Lambda \in S^2 (A)$. If $a$ is $(L,R)$-invariant, then
\begin{eqnarray*}
\langle \xi, x\ast_P a_+(\eta) \rangle+\langle \eta,x\cdot_P a_+(\xi)  \rangle&=&0, \\
\langle \xi, x\circ_P a_+(\eta) \rangle-\langle \eta,\{x, a_+(\xi)\}_P  \rangle&=&0,
\quad \forall x \in P, \xi,\eta \in P^{\ast}.
\end{eqnarray*}
Moreover, the pre-Poisson algebra structure on $P^{\ast}$ given by \eqref{Pr-1}-\eqref{Pr-2} reduces to
\begin{eqnarray*}
    \xi \ast_{r} \eta&=& -(L^*_*(\Lambda_{+}\xi)+R^*_*(\Lambda_{+}\xi))\eta+R^*_*(\Lambda_{+}\eta)\xi\\
    \xi \circ_{r} \eta&=&
    (L^*_{\circ}(\Lambda_{+}\xi)-R^*_{\circ}(\Lambda_{+}\xi))\eta-R^*_{\circ}(\Lambda_{+}\eta)\xi,
    \quad \forall \xi,\eta \in A^{\ast}.
\end{eqnarray*}
\end{lem}

\begin{proof}
It follows from \cite[Lemma 4.6]{Wang} and \cite[Lemma 2.13]{WBLS}.
\end{proof}

\begin{thm}\label{heart}
Let $(P,\ast_P,\circ_P)$ be a pre-Poisson algebra and $r\in P\otimes P.$ Assume that the skew-symmetric part $a$ of $r$ is $(L,R)$-invariant. Then $r$ satisfies the pre-Poisson Yang-Baxter equation $Z(r)=S(r)=0$ if and only if $(P^{\ast},\ast_{r},\circ_{r})$ is a pre-Poisson algebra and the linear maps $r_{+},r_{-}:(P^{\ast},\ast_{r},\circ_{r}) \longrightarrow (P,\ast_P,\circ_P)$ are both pre-Poisson algebra homomorphisms.
\end{thm}

\begin{proof}
Since the skew-symmetric part $a$ of $r$ is $(L,R)$-invariant, by \cite[Theorem 4.7]{Wang}, we deduce that $r$ satisfies $Z(r)=0$ if and only if $(P^{\ast},\ast_{r})$ is a Zinbiel algebra and the linear maps $r_{+},r_{-}:(P^{\ast},\ast_{r}) \longrightarrow (P,\ast_P)$ are both pre-Poisson algebra homomorphisms. Similarly, by \cite[Theorem 2.14]{WBLS}, we deduce that $r$ satisfies the $S(r)=0$ if and only if $(P^{\ast},\circ_{r})$ is a pre-Lie algebra and the linear maps $r_{+},r_{-}:(P^{\ast},\circ_{r}) \longrightarrow (P,\circ_P)$ are both pre-Lie algebra homomorphisms. Therefore, $r$ satisfies the pre-Poisson Yang-Baxter equation $Z(r)=S(r)=0$ if and only if $(P^{\ast},\ast_{r},\circ_{r})$ is a pre-Poisson algebra and the linear maps $r_{+},r_{-}:(P^{\ast},\ast_{r},\circ_{r}) \longrightarrow (P,\ast_P,\circ_P)$ are both pre-Poisson algebra homomorphisms.
\end{proof}

Next we introduce the notion of relative Rota-Baxter operators on pre-Poisson algebras. We need to define actions of pre-Poisson algebras.

\begin{defi}
Let $(P,\ast_P,\circ_P)$ and $(Q,\ast_Q,\circ_Q)$ be two pre-Poisson algebras. An {\bf action} of $(P,\ast_P,\circ_P)$ on $(Q,\ast_Q,\circ_Q)$ consists of four linear maps $\rho,\mu,\theta,\gamma:P \to \gl(Q)$, such that $(Q;\rho,\mu,\theta,\gamma)$ is a representation of the pre-Poisson algebra $(P,\ast_P,\circ_P)$ and the following equations are satisfied:
\begin{eqnarray}
\label{action1} u\ast_{Q} (\mu(x)v)&=&\mu(x)(u\ast_Q v+v\ast_Q u);\\
\label{action2} u\ast_{Q} (\rho(x)v)&=&(\mu(x)u)\ast_Q v+(\rho(x)u)\ast_Q v;\\
\label{action3} \rho(x)(u\ast_Q v)&=&(\mu(x)u)\ast_Q v+(\rho(x)u)\ast_Q v;\\
\label{action4} \theta(x)(u\circ_Q v)&=& (\theta(x)u)\circ_Q v+u\circ_Q (\theta(x)v)-(\gamma(x)u)\circ_Q v;\\
\label{action5} \gamma(x)(u\circ_Q v)&=& u\circ_Q(\gamma(x)v)-v\circ_Q(\gamma(x)u)+\gamma(x)(v \circ_Q u);\\
\label{action6} \mu(x)(u\circ_Q v-v\circ_Q u)&=&u \circ_Q (\mu(x)v)-v \ast_Q (\gamma(x)u);\\
\label{action7} \theta(x)(u\ast_Q v) &=& (\theta(x)u-\gamma(x)u)\ast_Q v+u \ast_Q (\theta(x)v);\\
\label{action8} \rho(x)(u\circ_Q v) &=& (\theta(x)u-\gamma(x)u)\ast_Q v+u \circ_Q (\rho(x)v);\\
\label{action9} \gamma(x)(u\ast_Q v+v\ast_Q u) &=& u \ast_Q (\gamma(x)v)+v \ast_Q (\gamma(x)u);\\
\label{action10}\rho(x)(u\circ_Q v)&=& (\rho(x)u+\mu(x)u)\circ_Q v-u\ast_Q(\theta(x)v),
\end{eqnarray}
for all $x\in A,u,v\in B$.
\end{defi}

\begin{rmk}
Similar to representations of pre-Poisson algebras, an action $(\rho,\mu,\theta,\gamma)$ of $(P,\ast_P,\circ_P)$ on $(Q,\ast_Q,\circ_Q)$ has an equivalent characterization in terms of the pre-Poisson algebra structure on $P\oplus Q$ given by
\begin{eqnarray*}
(x+u)\ast_{(\rho,\mu)} (y+v)&=& x\ast_P y+\rho(x)v+\mu(y)u+u\ast_Q v;\\
(x+u)\circ_{(\theta,\gamma)} (y+v)&=& x\circ_P y+\theta(x)v+\gamma(y)u+u\circ_Q v,\quad \forall x,y\in P,u,v\in Q.
\end{eqnarray*}
\end{rmk}

\begin{ex}{\rm
Let $(P,\ast_P,\circ_P)$ be a pre-Poisson algebra. By a direct calculation, we can check that  $(P;L_\ast,R_\ast,L_\circ,R_\circ)$ is an action of $(P,\ast_P,\circ_P)$ on itself, which is called the {\bf regular action}.
}
\end{ex}

\begin{defi}
Let $(P,\ast_P,\circ_P)$ and $(Q,\ast_Q,\circ_Q)$ be two pre-Poisson algebras and $(\rho,\mu,\theta,\gamma)$ be an action of $(P,\ast_P,\circ_P)$ on $(Q,\ast_Q,\circ_Q)$. Then a linear map $T:Q\to P$ is called a {\bf relative Rota-Baxter operator of weight $\lambda$} with respect to the action $(\rho,\mu,\theta,\gamma)$ of $P$ on $Q$, if $T$ satisfies the following equations:
\begin{eqnarray}
\label{rRBO1}(Tu)\ast_{P} (Tv)&=&T(\rho(Tu)v+\mu(Tv)u+\lambda u\ast_Q v),\\
\label{rRBO2}(Tu)\circ_{P} (Tv)&=&T(\theta(Tu)v+\gamma(Tv)u+\lambda u\circ_Q v), \quad \forall u,v\in Q.
\end{eqnarray}
In particular, a relative Rota-Baxter operator on a pre-Poisson algebra  $(P,\ast_P,\circ_P)$ with respect to the regular action is called a {\bf Rota-Baxter operator of weight $\lambda$}.  A {\bf Rota-Baxter pre-Poisson algebra} $(P,\ast_P,\circ_P,\huaB)$ of weight $\lambda$ is a pre-Poisson algebra equipped with a Rota-Baxter operator $\huaB$ of weight $\lambda$.
\end{defi}


Let $(P,\ast_P,\circ_P,\huaB)$ be a Rota-Baxter pre-Poisson algebra of weight $\lambda$. Then there is a new pre-Poisson algebra structure $(\ast_{\huaB},\cdot_{\huaB})$ on $P$ defined by
\begin{eqnarray*}
x \ast_{\huaB} y&=& \huaB(x)\ast_P y+x \ast_P \huaB(y)+\lambda x\ast_P y;\\
x \circ_{\huaB} y&=& \huaB(x)\circ_P y+x \circ_P \huaB(y)+\lambda x\circ_P y.
\end{eqnarray*}
The pre-Poisson algebra $(P,\ast_{\huaB},\cdot_{\huaB})$ is called the {\bf descendent pre-Poisson algebra} and denoted by $P_{\huaB}$. Furthermore, $\huaB$ is a pre-Poisson algebra homomorphism from  $P_{\huaB}$ to $(P,\ast_P,\circ_P)$.

If $(P,P_r^*)$ is a coboundary pre-Poisson bialgebra induced by $r\in P\otimes P$, then by $I=r_+ - r_-$, the pre-Poisson algebra structure on $P^*$ given by \eqref{Pr-1}-\eqref{Pr-2} reduces to
\begin{eqnarray}
\label{rRBO-Pr1}\xi \ast_r \eta&=&-(L^*_* (r_+\xi)+R^*_* (r_+\xi))\eta+R^*_*(r_{+}\eta) \xi-R^*_*(I\eta) \xi,\\
\label{rRBO-Pr2}\xi \circ_r \eta&=& (L^*_{\circ}-R^*_{\circ})(r_{+}\xi) \eta-R^*_{\circ}(r_{+}\eta) \xi+R^*_{\circ}(I\eta) \xi,
\quad \forall \xi,\eta\in P^*.
\end{eqnarray}

Now we define new multiplications $(\ast_+,\circ_+)$ on $P^*$ as follows:
\begin{equation}
\xi \ast_+ \eta:=-R^*_*(I\eta) \xi, \quad \xi \circ_+ \eta:=R^*_{\circ}(I\eta) \xi, \quad \forall \xi,\eta\in P^*.
\end{equation}

\begin{pro}\label{ast+}
If the skew-symmetric part $a$ of $r$ is $(L,R)$-invariant, then $(P^*,\ast_+,\circ_+)$ is a pre-Poisson algebra.
\end{pro}

\begin{proof}
It follows from \cite[Proposition 4.12]{Wang} that $(P^*,\ast_+)$ is a Zinbiel algebra. By Proposition \ref{invariance2}, we have
\begin{eqnarray*}
&&\langle (\xi\circ_+ \eta)\circ_+ \zeta-\xi \circ_+ (\eta\circ_+ \zeta)-(\eta \circ_+ \xi)\circ_+ \zeta+\eta \circ_+ (\xi \circ_+ \zeta), x\rangle\\
&=& \langle R_{\circ}^* (I\zeta)R_{\circ}^* (I\eta) \xi-R^*_{\circ} (I(R^*_{\circ}(I \zeta)\eta))\xi- R_{\circ}^* (I\zeta)R_{\circ}^* (I\xi) \eta+R^*_{\circ} (I(R^*_{\circ}(I \zeta)\xi))\eta,x  \rangle\\
&=& \langle  \xi,(x \circ_P (I\zeta))\circ_P (I\eta)-x \circ_P ((I\eta)\circ_P(I\zeta)) -\{x\circ_P (I\zeta),I\eta\}_P + \{x,I\eta\}_P \circ_P (I\zeta) \rangle\\
&=& \langle  \xi, (I \eta) \circ_P (x\circ_P (I \zeta) )-x \circ_P ((I \eta)\circ_P (I \zeta) )+
(x \circ_P (I \eta))\circ_P (I \zeta)-((I \eta) \circ_P x )\circ_P (I \zeta)\\
&=& 0,
\end{eqnarray*}
which implies that $(P,\circ_+)$ is a pre-Lie algebra. Moreover, we have
\begin{eqnarray*}
&&\langle(\xi \circ_+ \eta-\eta \circ_+ \xi)\ast_+ \zeta-\xi \circ_+ (\eta \ast_+ \zeta) + \eta \ast_+ (\xi \circ_+ \zeta),x \rangle\\
&=& \langle R^*_*(I \zeta)(-R_{\circ}^*(I \eta)\xi+R_{\circ}^*(I \xi)\eta )+R^*_{\circ} (I(R^*_*(I\zeta)\eta))\xi-  R^*_* (I(R^*_{\circ}(I\zeta)\xi))\eta  \rangle,\\
&=& \langle \eta, (x \ast_P (I \zeta))\circ_P (I \xi)-\{x \ast_P (I \zeta), I \xi\}_P+
\{ x,I \xi\}_P \ast_P (I \zeta)-x \ast_P ((I \xi)\circ_P (I \zeta)) \rangle\\
&=& \langle \eta, (x \ast_P (I \zeta))\circ_P (I \xi)-\{x \ast_P (I \zeta), I \xi\}_P- (I \xi)\circ_P (x \ast_P (I \zeta))\\
&=& 0,
\end{eqnarray*}
which implies that $\xi \circ_+ \eta-\eta \circ_+ \xi)\ast_+ \zeta=\xi \circ_+ (\eta \ast_+ \zeta) - \eta \ast_+ (\xi \circ_+ \zeta)$ for all $\xi,\eta,\zeta\in P^*$. Similarly, we can also have $\xi \ast_+ \eta+\eta \ast_+ \xi)\circ_+ \zeta=\xi \ast_+ (\eta \circ_+ \zeta) + \eta \ast_+ (\xi \circ_+ \zeta)$. Therefore, $(P^*,\ast_+,\circ_+)$ is a pre-Poisson algebra.
\end{proof}

\begin{pro}
If the skew-symmetric part $a$ of $r$ is $(L,R)$-invariant, then $(P^{\ast};-L_{\ast}^{\ast}-R_{\ast}^{\ast},R_{\ast}^{\ast},L_{\circ}^{\ast}-R_{\circ}^{\ast},
-R_{\circ}^{\ast})$ is an action of the pre-Poisson algebra $(P,\ast_P,\circ_P)$ on the pre-Poisson algebra $(P^*,\ast_+,\circ_+)$ given in Proposition \ref{ast+}.
\end{pro}

\begin{proof}
Since the skew-symmetric part $a$ of $r$ is $(L,R)$-invariant, by a direct calculation, we can check that the representation $(P^{\ast};-L_{\ast}^{\ast}-R_{\ast}^{\ast},R_{\ast}^{\ast},L_{\circ}^{\ast}-R_{\circ}^{\ast},
-R_{\circ}^{\ast})$ satisfies the equations \eqref{action1}-\eqref{action10} by Proposition \ref{invariance2}. We omit the details. Therefore,  $(P^{\ast};-L_{\ast}^{\ast}-R_{\ast}^{\ast},R_{\ast}^{\ast},L_{\circ}^{\ast}-R_{\circ}^{\ast},
-R_{\circ}^{\ast})$ is an action of the pre-Poisson algebra $(P,\ast_P,\circ_P)$ on the pre-Poisson algebra $(P^*,\ast_+,\circ_+)$.
\end{proof}

Then we show that quasi-triangular pre-Poisson bialgebras can give rise to relative Rota-Baxter operators of weight $1$ with respect to the action $(P^{\ast};-L_{\ast}^{\ast}-R_{\ast}^{\ast},R_{\ast}^{\ast},L_{\circ}^{\ast}-R_{\circ}^{\ast},
-R_{\circ}^{\ast})$.

\begin{thm}\label{quasi-rRBO}
Let $(P,P_{r}^{\ast})$  be a quasi-triangular pre-Poisson bialgebra induced by $r\in P \otimes P.$ Then $r_{+}:(P^*,\ast_+,\circ_+) \to (P,\ast_P,\circ_P) $ is a relative Rota-Baxter operator of weight $1$ with respect to the action $(P^{\ast};-L_{\ast}^{\ast}-R_{\ast}^{\ast},R_{\ast}^{\ast},L_{\circ}^{\ast}-R_{\circ}^{\ast},
-R_{\circ}^{\ast})$.
\end{thm}

\begin{proof}
By Theorem \ref{heart}, since $(P,P_{r}^{\ast})$ is a quasi-triangular pre-Poisson bialgebra induced by $r$, $r_{+}:(P^{\ast},\ast_{r},\circ_{r}) \longrightarrow (P,\ast_P,\circ_P)$ is a pre-Poisson algebra homomorphism, which implies that
$$ r_+(\xi\ast_{r} \eta)=(r_+\xi)\ast_P (r_+\eta),\quad  r_+(\xi\circ_{r} \eta)=(r_+\xi)\circ_P (r_+\eta), \quad \forall \xi,\eta\in P^*.$$
Then by \eqref{rRBO-Pr1}-\eqref{rRBO-Pr2}, we have
\begin{eqnarray*}
(r_+\xi)\ast_P (r_+\eta)=r_+(\xi\ast_{r} \eta)
&=& r_+(-(L_*^*+R_*^*)(r_+ \xi)\eta+R_*^*(r_+ \eta) \xi+\xi\ast_+ \eta);\\
(r_+\xi)\circ_P (r_+\eta)=r_+(\xi\circ_{r} \eta)
&=& r_+((L_{\circ}^*-R_{\circ}^*)(r_+ \xi)\eta-R_{\circ}^*(r_+ \eta) \xi+\xi\circ_+ \eta),
\end{eqnarray*}
which implies that $r_+:(P^*,\ast_+,\circ_+) \to (P,\ast_P,\circ_P)$ is a relative Rota-Baxter operators of weight $1$ with respect to the action $(P^{\ast};-L_{\ast}^{\ast}-R_{\ast}^{\ast},R_{\ast}^{\ast},L_{\circ}^{\ast}-R_{\circ}^{\ast},
-R_{\circ}^{\ast})$.
\end{proof}

\begin{rmk}
In fact, if we define a new multiplication on $P^*$ as follows:
\begin{equation*}
\xi \ast_- \eta:=-(L^*_*(I\xi)+R^*_*(I\xi)) \eta, \quad \xi \circ_- \eta:=(L^*_{\circ}-R^*_{\circ})(I\xi) \eta, \quad \forall \xi,\eta\in P^*,
\end{equation*}
then $(P^*,\ast_-,\circ_-)$ is also a pre-Poisson algebra. We can also prove that $r_-:(P^*,\ast_-,\circ_-) \to  (P,\ast_P,\circ_P)$ is a relative Rota-Baxter operator of weight $1$ with respect to the action $(P^{\ast};-L_{\ast}^{\ast}-R_{\ast}^{\ast},R_{\ast}^{\ast},L_{\circ}^{\ast}-R_{\circ}^{\ast},
-R_{\circ}^{\ast})$.
\end{rmk}

Taking a symmetric $r\in P\otimes P$ into Theorem \ref{quasi-rRBO}, we can obtain that a triangular pre-Poisson bialgebra can give rise to a relative Rota-Baxter operators of weight $0$ (also called $\huaO$-operator), with respect to the coregular representation $(P^{\ast};-L_{\ast}^{\ast}-R_{\ast}^{\ast},R_{\ast}^{\ast},L_{\circ}^{\ast}-R_{\circ}^{\ast},
-R_{\circ}^{\ast})$.

\begin{cor}
Let $(P,P_r^*)$ be a triangular pre-Poisson bialgebra induced by a symmetric $r\in P\otimes P$. Then $r_+,r_-:P^* \to P$ are relative Rota-Baxter operators of weight $0$ with
respect to the coregular representation $(P^{\ast};-L_{\ast}^{\ast}-R_{\ast}^{\ast},R_{\ast}^{\ast},L_{\circ}^{\ast}-R_{\circ}^{\ast},
-R_{\circ}^{\ast})$.
\end{cor}

\subsection{Factorizable pre-Poisson bialgebras}

\begin{defi}
A quasi-triangular pre-Poisson  bialgebra $(P,P_{r}^{\ast})$ is called {\bf factorizable} if the skew-symmetric part $a$ of $r$ is non-degenerate, which means that the linear map $I:P^{\ast} \longrightarrow P$ defined in \eqref{I} is a linear isomorphism of vector spaces.
\end{defi}

Consider the map
$$
P^*\stackrel{r_+\oplus r_-}{\longrightarrow}P\oplus P\stackrel{(x,y)\longmapsto x-y}{\longrightarrow}P.
$$
The following result justifies the terminology of a factorizable pre-Poisson bialgebra.
\begin{pro}
  Let  $(P,P_{r}^{\ast})$ be a factorizable pre-Poisson bialgebra. Then $\Img(r_+\oplus r_-)$ is a pre-Poisson subalgebra of the direct sum pre-Poisson algebra $P \oplus P$, which is isomorphic to the pre-Poisson algebra $(P^{\ast},\ast_{r},\circ_{r}).$ Moreover, any $x\in P$ has a unique decomposition
  \begin{equation}
    x=x_+-x_-,
  \end{equation}
  where $(x_+,x_-)\in \Img(r_+\oplus r_-)$.
\end{pro}
\begin{proof}
By Theorem \ref{heart}, both $r_+$ and $r_-$ are pre-Poisson algebra homomorphisms. Therefore, $\Img(r_+\oplus r_-)$ is a pre-Poisson subalgebra of the direct sum pre-Poisson algebra $P \oplus P$. Since $I=r_+  -r_-:P^*\to P$ is nondegenerate, it follows that the pre-Poisson algebra $\Img(r_+\oplus r_-)$ is isomorphic to the pre-Poisson algebra $(P^{\ast},\ast_{r},\circ_{r}).$

  Since $I:P^*\to P$ is nondegenerate, we have
  $$
 r_+ I^{-1}(x)-r_-I^{-1}(x)= (r_+  -r_-)I^{-1}(x)=x,
  $$
  which implies that $x=x_+-x_-,$ where $x_+=r_+ I^{-1}(x)$ and $x_-=r_-I^{-1}(x)$. The uniqueness also follows from the fact that $I:P^*\to P$ is nondegenerate.
\end{proof}


Let $(P,P^*)$ be a pre-Poisson bialgebra. Define multiplications $\ast_{\frkd}$ and $\circ_{\frkd}$ on $\frkd=P\oplus P^*$ by
\begin{eqnarray}
(x,\xi)\ast_{\frkd}(y,\eta)&=&\Big(  x\ast_P y-(\huaL^*_*+\huaR^*_*)(\xi)y+\huaR^*_*(\eta)x,\xi\ast_{P^*} \eta-(L_*^*+R_*^*)(x)\eta+R_*^*(y)\xi \Big); \\
(x,\xi)\circ_{\frkd} (y,\eta)&=&\big(x\circ_{P} y+(\huaL^*_{\circ}-\huaR^*_{\circ})({\xi}) y-\huaR^{\ast}_{\circ}({\eta})x ,\xi\circ_{P^{\ast}}\eta+ (L^*_{\circ}-R^*_{\circ})(x)\eta-R^{\ast}_{\circ}(y)\xi \big),
\end{eqnarray}
for all $x,y\in P,\xi,\eta\in P^*$. 
By Theorem \ref{pre-Poisson-bialgs}, $(\frkd,\cdot_{\frkd})$ is a pre-Poisson algebra, which is called the {\bf  double} of the  pre-Poisson bialgebra $(P,P^*).$

\begin{thm}\label{prePoisson-double-fac}
Let $(P,P^*)$ be a pre-Poisson bialgebra. Suppose that $\{ e_1,e_2,\dots,e_n\}$ is a basis of $P$ and $\{ e_1^*,e_2^*,\dots,e_n^*\}$ is the dual basis of $P^*$. Then $r=\sum_{i} e_i\otimes e_i^* \in P\otimes P^*  \subset \frkd \otimes \frkd$ induces a pre-Poisson algebra structure on $\frkd^*=P^*\oplus P$ such that $(\frkd,\frkd_r^*)$ is a quasi-triangular pre-Poisson bialgebra.
Moreover, $(\frkd,\frkd_r^*)$ is also a factorizable pre-Poisson bialgebra.
\end{thm}

\begin{proof}
By \cite[Theorem 4.20]{Wang} and \cite[Theorem 2.18]{WBLS}, we deduce that $Z(r)=0$ and the skew-symmetric part $a=\frac{1}{2}(e_i\otimes e_i^*-e_i^* \otimes e_i)$ of $r$ satisfies
\begin{eqnarray*}
a^{\sharp}(\tilde{L}^*_*(x,\xi)(\eta,y))+(x,\xi)\cdot_{\frkd} a^{\sharp}(\eta,y)=0;\\
a^{\sharp}(\tilde{L}^*_{\circ}(x,\xi)(\eta,y))-\{(x,\xi), a^{\sharp}(\eta,y)\}_{\frkd}=0;
\end{eqnarray*}
where $\tilde{L}_*$ and $\tilde{L}_{\circ}$ are the left multiplication operators of the Zinbiel algebra $(\frkd,\ast_{\frkd})$ and the pre-Lie algebra $(\frkd,\circ_{\frkd})$ respectively.
Thus, skew-symmetric part $a=\frac{1}{2}(e_i\otimes e_i^*-e_i^* \otimes e_i)$ of $r$ is $(\tilde{L},\tilde{R})$-invariant.
By \cite[Theorem 5.5]{Bai} and \cite[Theorem 2.18]{WBLS}, we deduce that $r$ satisfies the pre-Poisson equation $Z(r)=S(r)=0$. Thus, $(\frkd,\frkd_r^*)$ is a quasi-triangular pre-Poisson bialgebra.

Moreover, note that $r_+,r_{-}:\frkd_r^* \to \frkd$ are given respectively by
$$ r_+(\xi,x)=(0,\xi),\quad r_{-}(\xi,x)=(x,0), \quad \forall x\in P,\xi\in P^*. $$
This implies that $I(\xi,x)=(-x,\xi),$ which means that the linear map $I:\frkd_r^* \to \frkd$ is a linear isomorphism of vector spaces. Therefore, $(\frkd,\frkd_r^*)$ is also a factorizable pre-Poisson bialgebra.
\end{proof}

\section{Rota-Baxter characterization of factorizable pre-Poisson bialgebras}\label{sec:RB-fac}

In this section, we first add Rota-Baxter operators on quadratic pre-Poisson algebras and symplectic Poisson algebras, and show that there is still a one-to-one correspondence between quadratic Rota-Baxter pre-Poisson algebras and Rota-Baxter symplectic Poisson algebras. Finally we show that there is a one-to-one correspondence between factorizable pre-Poisson bialgebras and  quadratic Rota-Baxter pre-Poisson algebras, as well as Rota-Baxter symplectic Poisson algebras. Consequently, a Rota-Baxter symplectic Poisson algebra of weight $\lambda$ admits a different phase space than the semidirect product given in Theorem \ref{phasespace-subadj}.

First we add Rota-Baxter operators on quadratic pre-Poisson algebras and symplectic Poisson algebras, and introduce the notions of quadratic Rota-Baxter pre-Poisson algebras and Rota-Baxter symplectic Poisson algebras.

\begin{defi}\label{quadratic RB pre-Poisson algebra}
    Let$(P,\ast_P,\circ_P,\omega)$  a quadratic pre-Poisson algebra and $(P,\ast_P,\circ_P,\huaB)$ be a Rota-Baxter pre-Poisson algebra of weight $\lambda$. Then the quintuple $(P,\ast_P,\circ_P,\huaB,\omega)$ is called a {\bf quadratic Rota-Baxter pre-Poisson algebra of weight $\lambda$} if the following compatibility condition holds:
    \begin{eqnarray}\label{compatibility condition}
        \omega(\huaB x,y)+   \omega(x,\huaB y)+\lambda \omega(x,y)=0, \quad \forall x,y \in P.
    \end{eqnarray}
\end{defi}
\begin{defi}
    Let $(P,\cdot_P,\{\cdot,\cdot\}_P,\omega)$ be a symplectic Poisson algebra and $(P,\cdot_P,\{\cdot,\cdot\}_P,\huaB)$  be a Rota-Baxter Poisson algebra of weight $\lambda$. Then $(P,\cdot_P,\{\cdot,\cdot\}_P,\huaB,\omega)$ is  called a {\bf Rota-Baxter symplectic Poisson algebra of weight $\lambda$} if $\huaB$ and $\omega$ satisfy the compatibility condition \eqref{compatibility condition}.
\end{defi}

We now show that the relation between quadratic pre-Poisson algebras and symplectic Poisson algebras given in Theorem \ref{phase-space-Manin-triple} also holds for quadratic Rota-Baxter pre-Poisson algebras and Rota-Baxter symplectic Poisson algebras of the same weight.

\begin{thm}\label{thm:pre-Poisson-RB}
Let $(P,\ast_P,\circ_P,\huaB,\omega)$  be a quadratic Rota-Baxter pre-Poisson algebra of weight $\lambda$. Then  $(P^c,\cdot_P,\{\cdot,\cdot\}_P,\huaB,\omega)$ is a Rota-Baxter symplectic Poisson algebra of weight $\lambda$, where $\cdot_P,\{\cdot,\cdot\}_P$ is given by \eqref{sub-adj-poisson}.

Conversely, let $(P,\cdot_P,\{\cdot,\cdot\}_P,\huaB,\omega)$ be a Rota-Baxter symplectic Poisson algebra of weight $\lambda$. Then $(P,\ast_P,\circ_P,\huaB,\omega)$ is a quadratic Rota-Baxter pre-Poisson algebra of weight $\lambda$, where the pre-Poisson algebra structure $(P,\ast_P,\circ_P)$ is given by \eqref{Poisson-pre1} and \eqref{Poisson-pre2}.
\end{thm}

\begin{proof}
It is straightforward to deduce that if $\huaB:P \rightarrow P$ is a  Rota-Baxter operator of weight $\lambda$ on a pre-Poisson algebra $(P,\ast_P,\circ_P)$, then $\huaB$  is a  Rota-Baxter operator of weight $\lambda$ on the sub-adjacent Poisson algebra $(P^c,\cdot_P,\{\cdot,\cdot\}_P)$. Since  $(P,\ast_P,\circ_P,\omega)$ is a quadratic pre-Poisson algebra, by Theorem \ref{phase-space-Manin-triple}, $(P,\cdot_P,\{\cdot,\cdot\}_P,\omega)$ is a symplectic Poisson algebra. Therefore, if $(P,\ast_P,\circ_P,\huaB,\omega)$  is a quadratic Rota-Baxter pre-Poisson algebra of weight $\lambda$, then $(P,\cdot_P,\{\cdot,\cdot\}_P,\huaB,\omega)$ is a Rota-Baxter symplectic Poisson algebra of weight $\lambda$.

Conversely, let $(P,\cdot_P,\{\cdot,\cdot\}_P,\huaB,\omega)$ be a Rota-Baxter symplectic Poisson algebra of weight $\lambda$. First by Theorem \ref{phase-space-Manin-triple}, $(P,\ast_P,\circ_P,\omega)$ is a quadratic pre-Poisson algebra. By \cite[Theorem 5.4]{Wang} and \cite[Theorem 3.5]{WBLS}, we can obtain that $\huaB$ is a Rota-Baxter operator of weight $\lambda$ on the pre-Poisson algebra $(P,\ast_P,\circ_P)$. Therefore, $(P,\ast_P,\circ_P,\huaB,\omega)$ is a quadratic Rota-Baxter pre-Poisson algebra of weight $\lambda$.
\end{proof}

The following theorem shows that a factorizable pre-Poisson bialgebra gives a quadratic Rota-Baxter pre-Poisson algebra.

\begin{thm}\label{Factorizable-pre-Poisson-algebra}
Let $(P,P_{r}^{\ast})$ be a factorizable pre-Poisson bialgebra with $I=r_{+}-r_{-}$. Then $(P,\huaB,\omega_{I})$ is a quadratic Rota-Baxter pre-Poisson algebra of weight $\lambda$, where the linear map $\huaB:P \to P$ and $\omega_{I}\in \wedge^{2} P^{\ast}$ are defined respectively by
\begin{eqnarray}
\label{huaB} \huaB&=&\lambda r_{-}\circ I^{-1},\\
\label{omegaI}\omega_I(x,y)&=&\langle I^{-1}x,y \rangle, \quad \forall x,y \in P.
\end{eqnarray}
\end{thm}

\begin{proof}
Since $r_{+},r_{-}:(P^{\ast},\ast_{r},\cdot_{r}) \longrightarrow (P,\ast_P,\circ_P)$ are both pre-Poisson algebra homomorphisms, for all $x,y\in P$, we have
\begin{eqnarray}\label{RB operator in fac}
I(I^{-1}x \ast_{r} I^{-1}y )&=&(r_{+}-r_{-})(I^{-1}x \ast_{r} I^{-1}y )\\
\nonumber&=&((I+r_{-})I^{-1}x) \ast_P ((I+r_{-})I^{-1}y)- (r_{-}I^{-1}x)\ast_P (r_{-}I^{-1}y)\\
\nonumber&=& (r_{-}I^{-1}x) \ast_P y+x \ast_P (r_{-}I^{-1}y)+x \ast_P y.
\end{eqnarray}
Therefore, we have
\begin{eqnarray*}
 \huaB( \huaB(x) \ast_P y+x\ast_P  \huaB(y) +\lambda x\ast_P y)&=& \lambda^{2} r_{-}I^{-1}\Big((r_{-}I^{-1}x) \ast_P y+x \ast_P (r_{-}I^{-1}y)+x \ast_P y \Big)\\
&=& \lambda^{2}r_{-}(I^{-1}x \ast_{r} I^{-1}y )\\
&=& \lambda^{2}(r_{-}I^{-1}x \ast_P r_{-}I^{-1}y )\\
&=&  \huaB(x)\ast_P  \huaB(y),
\end{eqnarray*}
which implies that $\huaB$ is a Rota-Baxter operator of weight $\lambda$ on the Zinbiel algebra  $(P,\ast)$. Similarly, we have
$$ \huaB( \huaB(x) \circ_P y+x\circ_P  \huaB(y) +\lambda x\circ_P y)=\huaB(x)\circ_P  \huaB(y).$$
Thus, $\huaB$ is a Rota-Baxter operator of weight $\lambda$ on the pre-Poisson algebra $(P,\ast_P,\circ_P)$.

Next we show that $(P,\huaB,\omega_{I})$ is a quadratic Rota-Baxter pre-Poisson algebra. Since $I^{\ast}=-I$, we have
$$ \omega_{I}(x,y)=\langle I^{-1}x,y \rangle=-\langle x,I^{-1}y \rangle=-\omega_{I}(y,x),$$
 which means that $ \omega_{I}$ is skew-symmetric.

Since the skew-symmetric part $a$ of $r$ is $(L,R)$-invariant, by Proposition \ref{invariance2},
we have $I\circ L^*_*(x)=-(L_*(x)+R_*(x))\circ I$ and $I \circ L^*_{\circ}(x)=(L_\circ(x)-R_\circ(x))\circ I$. Thus, we have
\begin{eqnarray*}
\omega_{I}(x\ast_P y,z)-\omega_{I}(y,x\ast_P z+z\ast_P x)
&=&\langle I^{-1}(x\ast_P y),z  \rangle-\langle I^{-1}(y),x\ast_P z+z \ast_P x \rangle\\
&=&\langle I^{-1}\circ L_{\ast}(x)(y)+(L_*^*+R_*^*)(x) \circ I^{-1}(y),z \rangle\\
&=&0,
\end{eqnarray*}
and
\begin{eqnarray*}
\omega_{I}(x\circ_P y,z)+\omega_{I}(y,x\circ_P z-z\circ_P x)
&=&\langle I^{-1}(x\circ_P y),z  \rangle+\langle I^{-1}(y),x\circ_P z-z \circ_P x \rangle\\
&=&\langle I^{-1}\circ L_{\circ}(x)(y)-(L_{\circ}^*-R_{\circ}^*)(x) \circ I^{-1}(y),z \rangle\\
&=&0,
\end{eqnarray*}
which implies that \eqref{quadratic1}-\eqref{quadratic2} hold.

Moreover, by using $r^{\ast}_{-}=r_{+}$ and $I=r_{+}-r_{-}$, we have
\begin{eqnarray*}
\omega_{I}(x,\huaB y)+\omega_{I}(\huaB x,y)+\lambda \omega_{I}(x,y)&=&\lambda\Big(\langle I^{-1}(x), r_{-}I^{-1}(y) \rangle+\langle  I^{-1} r_{-} I^{-1}(x),y \rangle+\langle I^{-1}x,y \rangle \Big)\\
&=& \lambda\langle(-I^{-1} r_{+} I^{-1}+I^{-1} r_{-} I^{-1}+ I^{-1})(x),y \rangle\\
&=& 0,
\end{eqnarray*}
which implies that \eqref{compatibility condition} holds.

Therefore, $(P,\huaB,\omega_{I})$ is a quadratic Rota-Baxter pre-Poisson algebra of weight $\lambda$.
\end{proof}

It is straightforward to check that if $\huaB:P \to P$ is a Rota-Baxter operator of weight $\lambda$ on a pre-Poisson algebra $(P,\ast_P,\circ_P)$, then
\begin{eqnarray}
\widetilde{\huaB}:=-\lambda\Id-\huaB
\end{eqnarray}
is also a Rota-Baxter operator of weight $\lambda$.

\begin{cor}
Let $(P,P_{r}^{\ast})$ be a factorizable pre-Poisson bialgebra with $I=r_{+}-r_{-}$. Then $(P,\widetilde{\huaB},\omega_{I})$ is also a quadratic Rota-Baxter pre-Poisson algebra of weight $\lambda$, where $\widetilde{\huaB}=-\lambda\Id-\huaB=-\lambda r_{+}\circ I^{-1}$ and $\omega_{I}\in \wedge^{2} P^{\ast}$ is defined by
\eqref{omegaI}.
\end{cor}

\begin{proof}
Since $(P,P_{r}^{\ast})$ is a factorizable pre-Poisson bialgebra,  by Theorem \ref{Factorizable-pre-Poisson-algebra},  $\huaB$ satisfies the compatibility  condition \eqref{compatibility condition}.  Thus we have
\begin{eqnarray*}
&&\omega_{I}(x,\widetilde{\huaB}y)+\omega_{I}(\widetilde{\huaB}x,y)+\lambda \omega_{I}(x,y)\\
&=& -\omega_{I}(x,\huaB y)-\lambda\omega_{I}(x,y)-\omega_{I}(\huaB x,y)-\lambda \omega_{I}(x,y)+\lambda \omega_{I}(x,y)\\
&=& -\omega_{I}(x,\huaB y)-\omega_{I}(\huaB x,y)-\lambda\omega_{I}(x,y)\\
&=& 0.
\end{eqnarray*}
This implies that $(P,\widetilde{\huaB},\omega_{I})$ is  a quadratic Rota-Baxter pre-Poisson algebra of weight $\lambda$.
\end{proof}

\begin{cor}\label{pre-Poisson bialgebra isomorphism}
Let $(P,P_{r}^{\ast})$ be a factorizable pre-Poisson bialgebra with $I=r_{+}-r_{-}$ and $\huaB=\lambda r_{-}\circ I^{-1}$  the induced Rota-Baxter operator of weight $\lambda$. Then $((P_{\huaB},\ast_{\huaB},\circ_{\huaB}),(P^{\ast},\ast_{I},\circ_{I}))$ is a pre-Poisson bialgebra, where
\begin{eqnarray}
\xi \ast_{I} \eta:&=& -\lambda I^{-1}\Big((\frac{1}{\lambda} I \xi)\ast_P(\frac{1}{\lambda} I \eta)\Big);\\
\xi \circ_{I} \eta:&=& -\lambda I^{-1}\Big((\frac{1}{\lambda} I \xi)\circ_P(\frac{1}{\lambda} I \eta)\Big),\quad \forall \xi,\eta\in P^{\ast},~\lambda\neq 0.
\end{eqnarray}
Moreover, $\frac{1}{\lambda} I:P^{\ast} \longrightarrow P$ gives a pre-Poisson bialgebra isomorphism from $(P_{r}^{\ast},P)$ to
$((P_{\huaB},\ast_{\huaB},\circ_{\huaB}),$ $(P^{\ast},\ast_{I},\circ_{I}))$.
\end{cor}

\begin{proof}
By Example \ref{dual pre-Poisson alg}, we know that if $(P,P_{r}^{\ast})$ is a pre-Poisson bialgebra, then $(P_{r}^{\ast},P)$ is also a pre-Poisson bialgebra.
First we show that $\frac{1}{\lambda} I:P_{r}^{\ast} \longrightarrow P_{\huaB}$ is a pre-Poisson algebra isomorphism. In fact, for any $\xi,\eta \in A^{\ast},$ taking $x=I\xi\in A$ and $y=I\eta \in A$, by \eqref{RB operator in fac}, we have
\begin{eqnarray*}
\frac{1}{\lambda}I (\xi\ast_{r}\eta)=\frac{1}{\lambda^{2}}( PI\xi\ast_P I\eta+I\xi \ast_P PI\eta   +\lambda I\xi \ast_P I\eta)&=& (\frac{1}{\lambda}I\xi)\ast_{\huaB} (\frac{1}{\lambda}I\eta);\\
\frac{1}{\lambda}I (\xi\circ_{r}\eta)=\frac{1}{\lambda^{2}}( PI\xi \circ_P I\eta+I\xi \circ_P PI\eta   +\lambda I\xi \circ_P I\eta)&=& (\frac{1}{\lambda}I\xi) \circ_{\huaB} (\frac{1}{\lambda}I\eta),\\
\end{eqnarray*}
which implies that $\frac{1}{\lambda} I$ is a pre-Poisson algebra isomorphism.

Since $(\frac{1}{\lambda}I)^{\ast}=-\frac{1}{\lambda}I$, we have
\begin{eqnarray*}
(\frac{1}{\lambda}I)^{\ast} (\xi \ast_{I} \eta)=\Big((-\frac{1}{\lambda} I \xi)\ast_P (-\frac{1}{\lambda} I \xi)\Big)=(\frac{1}{\lambda}I)^{\ast}(\xi)\ast_P (\frac{1}{\lambda}I)^{\ast}(\eta);\\
(\frac{1}{\lambda}I)^{\ast} (\xi \circ_{I} \eta)=\Big((-\frac{1}{\lambda} I \xi)\circ_P (-\frac{1}{\lambda} I \xi)\Big)=(\frac{1}{\lambda}I)^{\ast}(\xi)\circ_P (\frac{1}{\lambda}I)^{\ast}(\eta),\\
\end{eqnarray*}
which means that $(\frac{1}{\lambda}I)^{\ast}=-\frac{1}{\lambda}I:(P^{\ast},\ast_{I},\circ_{I})\longrightarrow (P,\ast_P,\circ_P)$ is also a pre-Poisson algebra isomorphism. Since $(P_{r}^{\ast},P)$ is a pre-Poisson  bialgebra, the pair $((P_{\huaB},\ast_{\huaB},\circ_{\huaB}),$ $(P^{\ast},\ast_{I},\circ_{I}))$ is also a pre-Poisson  bialgebra. Obviously, $\frac{1}{\lambda} I$ is a pre-Poisson  bialgebra isomorphism.
\end{proof}

By Theorems \ref{thm:pre-Poisson-RB} and   \ref{Factorizable-pre-Poisson-algebra}, we have
\begin{cor}
    Let $(P,P_{r}^{\ast})$ be a factorizable pre-Poisson bialgebra with $I=r_{+}-r_{-}$. Then $(P^c,\cdot_{P},\{\cdot,\cdot\}_{P},\huaB,\omega_I)$ is a Rota-Baxter symplectic Poisson algebra of weight $\lambda$, where the linear map $\huaB:P\longrightarrow P$ and $\omega_{I}\in \wedge^{2} P^{\ast}$ are defined by \eqref{huaB} and \eqref{omegaI}, respectively.
\end{cor}

\begin{ex}\label{pre-Poisson double}{\rm
Let $(\frkd,\frkd_r^{\ast})$ be a factorizable pre-Poisson bialgebra given in Theorem \ref{prePoisson-double-fac}.
By Theorem \ref{Factorizable-pre-Poisson-algebra},  $(\frkd,\cdot_{\frkd},\huaB,\omega_I)$ is a quadratic Rota-Baxter pre-Poisson algebra of weight $\lambda$,  where
\begin{eqnarray*}
\huaB(x,\xi)&=&\lambda r_{-}\circ I^{-1} (x,\xi)=-\lambda(x,0),\\
\omega_I(x+\xi,y+\eta)&=& \langle I^{-1}(x+\xi),y+\eta \rangle=\langle -x+\xi,y+\eta \rangle=\xi(y)-\eta(x),
\end{eqnarray*}
for all $x,y\in P, \xi,\eta\in P^*$. Note that
$\omega_I$ is exactly the bilinear form given by {\rm(\ref{natural-skew-sym-form})}.
Furthermore, $(\frkd,[\cdot,\cdot]_{\frkd},\huaB,\omega_I)$ is a Rota-Baxter symplectic Poisson algebra of weight $\lambda$, where the Poisson algebra structure $(\cdot_{\frkd},\{\cdot,\cdot\}_{\frkd})$  on $\frkd$ is given by
\begin{eqnarray*}
(x+\xi)\cdot_{\frkd} (y+\eta)&=&x \cdot_P y-\huaL^*_*(\xi) y-\huaL^*_*(\eta) x+
\xi \cdot_{P^*} \eta-L^{\ast}_* (x)\eta-L^{\ast}_* (y) \xi; \\
\{x+\xi,y+\eta\}_{\frkd}&=&\{x, y\}_P+\huaL^*_{\circ}(\xi) y-\huaL^*_{\circ}(\eta) x+
\{\xi, \eta\}_{P^*}+L^{\ast}_{\circ} (x)\eta-L^{\ast}_{\circ} (y) \xi,
\end{eqnarray*}
} for all $x,y\in P,\xi,\eta\in P^*.$
\end{ex}

At the end of this section, we show that a quadratic Rota-Baxter pre-Poisson algebra of nonzero weight can  give rise to a factorizable pre-Poisson bialgebra.

\begin{thm}\label{thm:QRB-facpre-Poisson}
    Let $(P,\huaB,\omega)$ be a quadratic Rota-Baxter pre-Poisson algebra of weight $\lambda$ $(\lambda\neq 0),$ and
    $ \huaI_{\omega}: P^{\ast}\longrightarrow P $ the induced linear isomorphism given by $\langle\huaI_{\omega}^{-1}x,y \rangle :=\omega(x,y).$ Then $r \in P\otimes P $ defined by
\begin{equation}\label{eq:equiv}
  r_{+}:=\frac{1}{\lambda} (\huaB+\lambda
\Id)\circ \huaI_{\omega}:P^{\ast} \longrightarrow P, \quad
r_{+}(\xi)=r(\xi,\cdot), \quad \forall \xi\in P^{\ast}
\end{equation}
    satisfies the pre-Poisson Yang-Baxter equation $Z(r)=S(r)=0$ and thus gives rise to a factorizable pre-Poisson bialgebra $(P,P_{r}^{\ast})$.
\end{thm}

\begin{proof}
    Since $\omega$ is skew-symmetric, we have $\huaI_{\omega}=-\huaI_{\omega}^{\ast}$. By the fact that $\omega(x,\huaB y)+\omega(\huaB x,y)+\lambda \omega(x,y)=0$  for all $x,y\in P$, we have
    $$ \langle \huaI_{\omega}^{-1}x,\huaB (y)\rangle+\langle \huaI_{\omega}^{-1}\circ \huaB(x),y \rangle+\lambda\langle \huaI_{\omega}^{-1}x,y \rangle=0, $$
    which implies that $\huaB^{\ast}\circ \huaI_{\omega}^{-1}+\huaI_{\omega}^{-1}\circ \huaB+\lambda \huaI_{\omega}^{-1}=0, $ and then
    $$ \huaI_{\omega}\circ \huaB^{\ast}+\huaB \circ \huaI_{\omega}+\lambda \huaI_{\omega}=0.$$
Thus we have
    $$ r_{-}:=r_{+}^{\ast}=\frac{1}{\lambda}(-\huaI_{\omega}\circ \huaB^{\ast}-\lambda \huaI_{\omega})=\frac{1}{\lambda} \huaB \circ \huaI_{\omega}, $$
    and $\huaI_{\omega}=r_{+}-r_{-}$. Define multiplications $(\ast_r,\circ_r)$ on $P^{\ast}$ by
    \begin{eqnarray*}
    \xi \ast_r \eta&=& -(L_*^*+R_*^*)(r_+ \xi)\eta+R^*_* (r_{-} \eta)\xi;\\
    \xi \circ_r \eta&=& (L_{\circ}^*-R_{\circ}^*)(r_+ \xi)\eta-R^*_{\circ} (r_{-} \eta)\xi.
    \end{eqnarray*}
Now we show that the following equation holds:
\begin{eqnarray}
\label{pre-Poisson alg iso1}   \frac{1}{\lambda} \huaI_{\omega}(\xi \ast_{r} \eta)=(\frac{1}{\lambda}\huaI_{\omega}\xi)\ast_{\huaB} (\frac{1}{\lambda}\huaI_{\omega}\eta);\\
\label{pre-Poisson alg iso2}   \frac{1}{\lambda} \huaI_{\omega}(\xi \circ_{r} \eta)=(\frac{1}{\lambda}\huaI_{\omega}\xi)\circ_{\huaB} (\frac{1}{\lambda}\huaI_{\omega}\eta).
\end{eqnarray}
By the fact that $\omega(x\ast_P y,z)-\omega(y,x \ast_P z+z \ast_P x)=0$ and $\omega(x\circ_P y,z)+\omega(y,x \circ_P z-z \circ_P x)=0$ for all $x,y,z\in P$, we have
\begin{eqnarray*}
\langle \huaI_{\omega}^{-1}\circ L_* (x)y,z \rangle+\langle (L_*^*+R_*^*)(x)\circ \huaI_{\omega}^{-1}(y),z \rangle=0;\\
\langle \huaI_{\omega}^{-1}\circ L_{\circ}(x)y,z \rangle-\langle (L_{\circ}^*-R_{\circ}^*)(x)\circ \huaI_{\omega}^{-1}(y),z \rangle=0,
\end{eqnarray*}
which implies that  $ L_\ast (x)\circ \huaI_{\omega}=-\huaI_{\omega} \circ (L_*^*+R_*^*)(x)$ and
$ L_{\circ} (x)\circ \huaI_{\omega}=\huaI_{\omega} \circ (L_{\circ}^*-R_{\circ}^*)(x).$
Therefore, by Proposition
\ref{invariance2}, the skew-symmetric part $a$ of $r$ is $(L,R)$-invariant.

On the one hand, for all $\xi,\eta\in P^*$, we have
\begin{eqnarray*}
    \huaI_{\omega}(\xi\ast_r \eta)&=&\huaI_{\omega}\Big(-(L_*^*+R_*^*)(r_+ \xi)\eta+R_*^*(r_{-} \eta) \xi \Big)\\
    &=& L_* (r_{+} \xi) \circ \huaI_{\omega}(\eta)+\huaI_{\omega}\circ R_*^{\ast}(r_{-} \eta)(\xi)\\
    &=& L_* (r_{+} \xi) \circ \huaI_{\omega}(\eta)+R_* (r_{-} \eta) \circ \huaI_{\omega}(\xi)\\
    &=&r_+(\xi)\ast_P (r_{+}(\eta)-r_{-}(\eta))+(r_{+}(\xi)-r_{-}(\xi))\ast_P r_-(\eta) \\
    &=& r_{+}(\xi)\ast_P r_{+}(\eta)-r_{-}(\xi)\ast_P r_{-}(\eta).
\end{eqnarray*}
On the other hand, we have
\begin{eqnarray*}
    &&(\huaI_{\omega}\xi)\ast_{\huaB}(\huaI_{\omega}\eta)\\
    &=&\huaB(\huaI_{\omega}\xi)\ast_P(\huaI_{\omega}\eta)+(\huaI_{\omega}\xi)\ast_P \huaB(\huaI_{\omega}\eta)
    +\lambda(\huaI_{\omega}\xi)\ast_P(\huaI_{\omega}\eta)\\
    &=& \lambda (r_{-}\xi)\ast_P(r_{+}\eta-r_{-}\eta)+\lambda(r_{+}\xi-r_{-}\xi)\ast_P (r_{-}\eta)+\lambda(r_{+}\xi-r_{-}\xi)\ast_P(r_{+}\eta-r_{-}\eta)\\
    &=& \lambda (r_{+}\xi)\ast_P (r_{+}\eta)-\lambda (r_{-}\xi)\ast_P (r_{-}\eta),
\end{eqnarray*}
which implies that  \eqref{pre-Poisson alg iso1} holds. Similarly, we can also verify \eqref{pre-Poisson alg iso2} holds. Thus $(P^*,\ast_{r},\circ_r)$ is a pre-Poisson algebra and $\frac{1}{\lambda} \huaI_{\omega}$ is a pre-Poisson algebra isomorphism from $(P^*,\ast_{r},\circ_r)$ to $(P,\ast_{\huaB},\circ_{\huaB})$.

 Finally, by the fact that $\huaB+\lambda\Id,\huaB:(P,\ast_{\huaB},\circ_{\huaB}) \longrightarrow (P,\ast_P,\circ_P)$ are both pre-Poisson algebra homomorphisms, we deduce that
$$   r_{+}:=\frac{1}{\lambda} (\huaB+\lambda \Id)\circ \huaI_{\omega},  \quad r_{-}:=\frac{1}{\lambda} \huaB\circ \huaI_{\omega}:(P^*,\ast_{r},\circ_r) \longrightarrow (P,\ast_P,\circ_P)$$
    are both pre-Poisson algebra homomorphisms. Therefore, by Theorem \ref{heart}, we have $Z(r)=S(r)=0$ and $(P,P_{r}^{\ast})$ is a quasi-triangular pre-Poisson bialgebra. Since $\huaI_{\omega} = r_{+}-r_{-}$  is an isomorphism, the pre-Poisson bialgebra $(P,P_{r}^{\ast})$ is factorizable.
\end{proof}

\begin{cor}\label{new-phase-space}
Let $(P,\cdot_P,\{\cdot,\cdot\}_P,\huaB,\omega)$ be a Rota-Baxter symplectic Poisson algebra of weight $\lambda$. Then $(P\oplus P^*,\cdot,\{\cdot,\cdot\},\omega)$ is a phase space of the Poisson algebra $(P,\cdot_P,\{\cdot,\cdot\}_P)$, where the Poisson structure $(\cdot,\{\cdot,\cdot\})$ on $P\oplus P^*$ given by
\begin{eqnarray*}
(x+\xi)\cdot(y+\eta)&=&x\cdot_P y-\huaL^*_*(\xi)y-\huaL^*_*(\eta)x+\xi\cdot_{P^*} \eta-L^*_*(x)\eta-L^*_*(y)\xi;\\
\{x+\xi,y+\eta\}&=&\{x, y\}_P+\huaL^*_{\circ}(\xi)y-\huaL^*_{\circ}(\eta)x+\{\xi,\eta\}_{P^*} +L^*_{\circ}(x)\eta-L^*_{\circ}(y)\xi.
\end{eqnarray*}
For all $x,y,z\in P$ and $\xi,\eta,\zeta\in P^*$.
\end{cor}

\begin{proof}
By Theorem \ref{thm:pre-Poisson-RB}, if $(P,\cdot_P,\{\cdot,\cdot\}_P,\huaB,\omega)$ is a Rota-Baxter symplectic Poisson algebra of weight $\lambda$, then $(P,\ast_P,\circ_P,\huaB,\omega)$ is a quadratic Rota-Baxter pre-Poisson algebra of weight $\lambda$, where the compatible pre-Poisson algebra structure $(P,\ast_P,\circ_P)$ is given by \eqref{Poisson-pre1} and \eqref{Poisson-pre2}. It is obvious that $(P\oplus P^*,P,P^*,\omega)$ is a Manin triple of pre-Poisson algebras. By Theorem \ref{phase-space-Manin-triple}, $(P\oplus P^*,\cdot,\{\cdot,\cdot\},\omega)$ is a phase space of the Poisson algebra $(P,\cdot_P,\{\cdot,\cdot\}_P)$, where the Poisson structure $(\cdot,\{\cdot,\cdot\})$ on $P\oplus P^*$ given by
\begin{eqnarray*}
(x+\xi)\cdot(y+\eta)&=&x\cdot_P y-\huaL^*_*(\xi)y-\huaL^*_*(\eta)x+\xi\cdot_{P^*} \eta-L^*_*(x)\eta-L^*_*(y)\xi;\\
\{x+\xi,y+\eta\}&=&\{x, y\}_P+\huaL^*_{\circ}(\xi)y-\huaL^*_{\circ}(\eta)x+\{\xi,\eta\}_{P^*} +L^*_{\circ}(x)\eta-L^*_{\circ}(y)\xi.
\end{eqnarray*}
For all $x,y,z\in P$ and $\xi,\eta,\zeta\in P^*$.
\end{proof}

\begin{rmk}
By Theorem \ref{phasespace-subadj}, a symplectic Poisson algebra $(P,\cdot_P,\{\cdot,\cdot\}_P,\omega)$ can give rise to a phase space $(P\ltimes_{-L^*_*,L^*_{\circ}} P^*,\cdot_{\ltimes},\{\cdot,\cdot\}_{\ltimes},\omega)$ of the Poisson algebra $(P,\cdot_P,\{\cdot,\cdot\}_P)$, which is different from the one in Corollary \ref{new-phase-space} according to the different Poisson algebra structures on $P\oplus P^*$.
\end{rmk}

\vspace{2mm}
\noindent
{\bf Acknowledgements.}    Supported by NSFC (12471060, W2412041).


\begin{thebibliography}{a}

\bibitem{Aguiar}
M. Aguiar, Pre-Poisson algebras. \emph{Lett. Math. Phys.} 54 (2000), 263-277.

\bibitem{Aguiar3}
M. Aguiar, On the associative analog of Lie bialgebras. \emph{J. Algebra} 244 (2001), 492-532.




\bibitem{Bai-phase-space}
C. Bai, Further study on non-abelian phase spaces: left-symmetric algebraic approach and related geometry. \emph{Rev. Math. Phys.} 18 (2006), no.5, 545-564.


\bibitem{Bai}
C. Bai, Left-symmetric bialgebras and an analogue of the classical Yang-Baxter equation. \emph{Commun. Contemp. Math.} 10 (2008), 221-260.

\bibitem{Bai-asso}
C. Bai, Double constructions of Frobenius algebras, Connes cocycles and their duality. \emph{J. Noncommut. Geom.} 4 (2010), no.4, 475-530.

\bibitem{Bai-Review} C. Bai, An introduction to pre-Lie algebras. In: Algebra and Applications 1: Nonssociative Algebras and Categories, Wiley Online Library (2021), 245-273.




\bibitem{BGN}
C. Bai, L. Guo and X. Ni, Generalizations of the classical Yang-Baxter equation and $\huaO$-operators. \emph{J. Math. Phys.} 52 (2011), no.6, 063515.


\bibitem{Ba}
G. Baxter, An analytic problem whose solution follows from a simple algebraic identity. \emph{Pacific J. Math.} {10}  (1960), 731-742.


\bibitem{Bu}
		D. Burde, Left-symmetric algebras, or pre-Lie algebras in geometry and physics. {\em Cent. Eur. J. Math.} {\bf 4} (2006), 323-357.


\bibitem{CK}
A. Connes and D. Kreimer, Renormalization in quantum field theory and the Riemann-Hilbert problem. I. The Hopf algebra structure of graphs and the main theorem. \emph{Comm. Math. Phys.} 210 (2000), 249-273.

\bibitem{Drinfeld}
V. G. Drinfeld, Hamiltonian structures on Lie groups, Lie bialgebras and the geometric meaning of the classical Yang-Baxter equations. \emph{Dokl. Akad. Nauk SSSR}, 268 (1983), no.2, 285-287.




\bibitem{G1}
M. Goncharov, On Rota-Baxter operators of non-zero weight arisen from the solutions of the classical Yang-Baxter equation. \emph{ Sib. El. Math. Rep.} 14 (2017), 1533-1544.

\bibitem{G2}
M. Goncharov, Rota-Baxter operators and non-skew-symmetric solutions of the classical Yang-Baxter equation on quadratic Lie algebra. \emph{Sib. El. Math. Rep.} 16 (2019), 2098-2109.

\bibitem{G3}
M. Goncharov and V. Gubarev, Double Lie algebras of nonzero weight. \emph{Adv. Math.} 409 (2022), 108680.

\bibitem{Goncharov}
M. E. Goncharov and P. S. Kolesnikov, Simple finite-dimensional double algebras. \emph{J. Algebra} 500 (2018), 425-438.

\bibitem{Guo}
L. Guo,  An introduction to Rota-Baxter algebra. Surveys of Modern Mathematics, 4. International Press, Somerville, MA; Higher Education Press, Beijing, 2012.

\bibitem{HB}
D. Hou and C. Bai, Some non-abelian phase spaces in low dimensions. \emph{J. Geom. Phys.}
58 (2008), no.12, 1752-1761.

\bibitem{Kosmann}
Y. Kosmann-Schwarzbach, Lie bialgebras, Poisson Lie groups and dressing transformation. In: Kosmann-Schwarzbach, K. M. Tamizhmani, B. Grammaticos (eds) Integrability of Nonlinear Systems. Lecture Notes in Physics, vol 638. Springer, Berlin, Heidelberg, 2004, 107-173.

\bibitem{Ku-phase-space}
B. A. Kupershmidt, Non-abelian phase spaces. \emph{J. Phys. A.} 27 (1994), no.8, 2801-2809.

\bibitem{Ku2}
B. A. Kupershmidt, On the nature of the Virasoro algebra. \emph{J. Nonlinear Math. Phy.} 6 (1999), 222-245.


\bibitem{Ku}
B. A. Kupershmidt, What a classical $r$-matrix really is. \emph{J. Nonlinear Math. Phys.} 6 (1999), 448-488.

\bibitem{Lang}
H. Lang and Y. Sheng, Factorizable Lie bialgebras, quadratic Rota-Baxter Lie algebras and Rota-Baxter Lie bialgebras. \emph{Comm. Math. Phys.} 397 (2023), 763-791.

\bibitem{LB}
G. Liu and C. Bai, New splittings of operations of Poisson algebras and transposed Poisson algebras and related algebraic structures. \emph{STEAM-H: Sci. Technol. Eng. Agric. Math. Health.} 2023, 49-96.

\bibitem{LMS}
S. Liu, A. Makhlouf and L. Song, On Hom-pre-Poisson algebras. \emph{J. Geom. Phys.} 190 (2023), 104855.



\bibitem{Liv}
M. Livernet, Rational homotopy of Leibniz algebras. \emph{Manuscripta Math.} 96 (1998), 295-315.

\bibitem{Loday}
J. L. Loday, Cup-product for Leibniz cohomology and dual Leibniz algebras. \emph{Math. Scand.} 77(1995), no.2, 189-196.


\bibitem{Lod1}
J.-L. Loday, Cup product for Leibniz cohomology and dual Leibniz algebras. \emph{Math. Scand.} 77, Univ. Louis Pasteur, Strasbourg, 1995, pp. 189-196.

\bibitem{Lod2}J.-L. Loday, Dialgebras. In Dialgebras and related operads. \emph{Lecture Notes in Math.} 1763, Springer, Berlin 2001, 7-66.

\bibitem{Mon}
S. Montgomery, Hopf algebras and their actions on rings, \emph{Amer. Math. Soc.}, Regional Conf. Ser. in Math., 82, 1993.

\bibitem{NB}
X. Ni and C. Bai, Poisson bialgebras. \emph{J. Math. Phys.} 54 (2013), no.2, 023515.

\bibitem{RS}
N. Reshetikhin and M. A. Semenov-Tian-Shansky, Quantum $R$-matrices and factorization problems. \emph{J. Geom. Phys.} 5 (1988), 533-550.

\bibitem{STS}
M. A. Semenov-Tian-Shansky, What is a classical $r$-matrix? \emph{Funct. Anal. Appl.} 17 (1983), 259-272.

\bibitem{S2}
M. A. Semenov-Tian-Shansky, Integrable systems and factorization problems. \emph{Operator Theory: Advances and Applications} 141 (2003), 155-218.


\bibitem{Wang}
Y. Wang, Zinbiel bialgebras, relative Rota-Baxter operators and the related Yang-Baxter equation.
arXiv.2504.15889

\bibitem{WBLS}
Y. Wang, C. Bai, J. Liu and Y. Sheng, Quasi-triangular pre-Lie bialgebras, factorizable pre-Lie bialgebras and Rota-Baxter pre-Lie algebras.  \emph{J. Geom. Phys.} 199 (2024), no. 105146.


\bibitem{Zhe}
V. N. Zhelyabin, Jordan bialgebras and their connection with Lie bialgebras. (Russian) \emph{Algebra i Logika} 36 (1997), no. 1, 3-25, 117; translation in \emph{Algebra and Logic} 36 (1997), no. 1, 1-15.

\end{thebibliography}
\end{document}